\documentclass{article}[12pt]
\usepackage[utf8]{inputenc}
\usepackage{amsmath,amssymb,amsthm}
\usepackage{graphicx}
\usepackage{authblk}
\usepackage{makecell}
\usepackage{hhline}
\usepackage{amsfonts,amssymb,array,bbm,bm,xcolor,float,footnote,dsfont,setspace,enumitem,graphicx,rotating,booktabs,layout,multirow,leftidx,url,tikz}
\usepackage{subcaption}
\usepackage{url}
\usepackage[
backend=biber, natbib,
style=numeric,
citestyle=authoryear,
sorting = nyt
]{biblatex}
\usepackage[ruled,vlined]{algorithm2e}

\newtheorem{thm}{Theorem}[section]

\newcommand{\R}{\mathbb{R}}

\newcommand{\df}{\;\mathrm{d}}
\newcommand{\code}[1]{\texttt{#1}}
\usepackage[left=1in,right=1in,top=1in,bottom=1in]{geometry}


\newcommand{\be}{\begin{eqnarray}}
\newcommand{\ee}{\end{eqnarray}}
\newcommand{\by}{\begin{eqnarray*}}
\newcommand{\ey}{\end{eqnarray*}}
\newcommand{\bn}{\begin{enumerate}}
\newcommand{\en}{\end{enumerate}}

\newlist{casesp}{enumerate}{3} 
\setlist[casesp]{align=left,
                 listparindent=\parindent,
                 parsep=\parskip,
                 font=\normalfont\bfseries,
                 leftmargin=0pt, 
                 labelwidth=0pt, 
                 itemindent=.4em,labelsep=.4em, 
                 partopsep=0pt, }

\setlist[casesp,1]{label=Case~\arabic*:,ref=\arabic*}
\setlist[casesp,2]{label=Case~\thecasespi.\arabic*:,ref=\thecasespi.\arabic*}
\setlist[casesp,3]{label=Case~\thecasespii.\alph*:,ref=\thecasespii.\alph*}

\usepackage[colorlinks]{hyperref}

\usepackage{setspace}
\usetikzlibrary{shapes,shadows,arrows,positioning,decorations.markings,decorations.pathreplacing,patterns}

\begin{document}

\title{\bf Pandemic Risk Management: Resources Contingency Planning and Allocation\footnote{First version: April 8, 2020.}}

\author{Xiaowei Chen\thanks{School of Finance, Nankai University. Email: chenx@nankai.edu.cn.}, Wing Fung Chong\thanks{Department of Mathematics and Department of Statistics, University of Illinois at Urbana-Champaign. Email: wfchong@illinois.edu.}, Runhuan Feng\thanks{Department of Mathematics, University of Illinois at Urbana-Champaign. Email: rfeng@illinois.edu.}, Linfeng Zhang\thanks{Department of Mathematics, University of Illinois at Urbana-Champaign. Email: lzhang18@illinois.edu.}}

\date{\today}
\maketitle



\begin{abstract}

Repeated history of pandemics, such as SARS, H1N1, Ebola, Zika, and COVID-19, has shown that pandemic risk is inevitable. Extraordinary shortages of medical resources have been observed in many parts of the world. Some attributing factors include the lack of sufficient stockpiles and the lack of coordinated efforts to deploy existing resources to the location of greatest needs. 

The paper investigates contingency planning and resources allocation from a risk management perspective, as opposed to the prevailing supply chain perspective. The key idea is that the competition of limited critical resources is not only present in different geographical locations but also at different stages of a pandemic. This paper draws on an analogy between risk aggregation and capital allocation in finance and pandemic resources planning and allocation for healthcare systems. The main contribution is to introduce new strategies for optimal stockpiling and allocation balancing spatio-temporal competitions of medical supply and demand.

\end{abstract}

{\allowdisplaybreaks

\section*{Introduction}

\subsection*{Lessons from Recent Pandemics}


An epidemic is an outbreak of a disease that spreads rapidly to a cohort of individuals in a wide area. According to the definition of the World Health Organization (WHO), ``a pandemic is the worldwide spread of a new disease." Because human has little immunity to the new disease, a pandemic can emerge quickly around the world. One of the most disastrous pandemics in the recent history is the 1918 flu pandemic which infected around 500 million and resulted in the deaths of estimated 50 million people worldwide, more than those died from the World War I. Most recently, the novel {\it coronavirus disease of 2019} (COVID-19) is an infectious disease caused by a new virus that only emerged in late 2019 and has since spread out to nearly every country in the world. According to Cutler and Summers, the COVID-19 could cause financial losses that sum up to \$16 trillion, or roughly $90\%$ of the annual GDP in the United States \citep{Cutler2020}. However, its far-reaching impact and consequent economic fallout due to the loss of productivity is yet to be realized; as of November 28, 2020, the number of infected cases has accumulated to over 61 million worldwide with the death toll reaching over 1.4 million (see, for example, \citet{Atlantic2020}).

The repeated history of pandemics in recent decades, such as SARS, swine flu, Ebola, and the most recent COVID-19, has taught us that pandemic risk is {\it inevitable}. Recent research studies (see, for instance, \citet{Jones2008} and \citet{Morse1995}) have shown that the frequency of pandemics has increased over the past century due to increased social connectivity, long-distance travel, urbanization, changes in land use, trade and consumption of wild animals, and greater exploitation of the natural environment. Figure \ref{fig:trend} visualizes both frequency and severity of well-recognized pandemics and public health emergencies of international concern declared by the WHO since 1900s. The vertical axis represents the number of documented infections on a logarithmic scale. Both the size and the color scale of the circles indicate the number of deaths resulting from the pandemics and public health emergencies. The alarming pattern of increased frequency clearly points to the critical importance of pandemic risk management.

\begin{figure}[h]
\centering
\includegraphics[width=\textwidth]{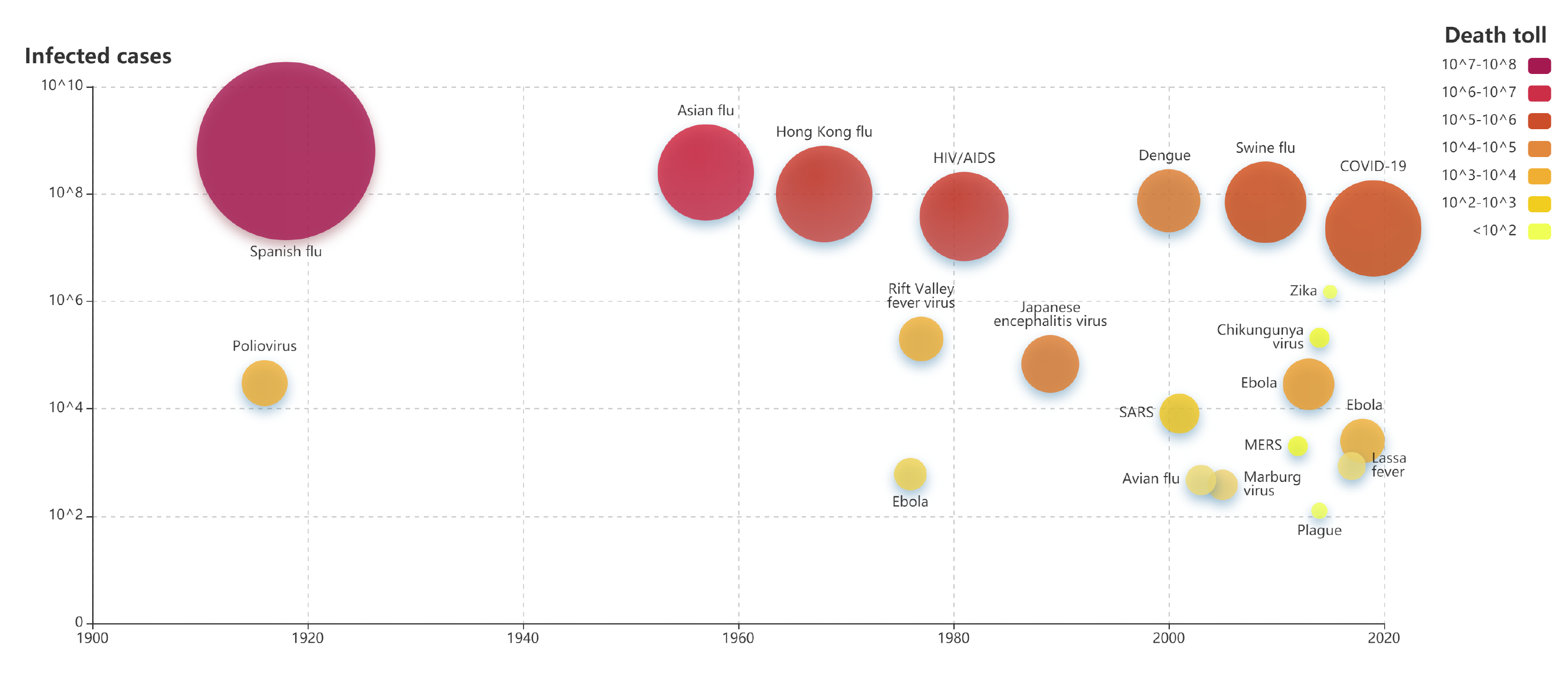}
\caption{Frequency and severity of pandemics and public health emergencies since 1900s}
\label{fig:trend}
\end{figure}


Governments around the world have been taking blames for their failure to promptly implement appropriate policies to contain the pandemic. Many countries experienced severe shortages of resources. \citet{Ranney2020} studied the critical role of scarce resources, such as ventilators and personal protective equipment (PPE), in shaping the direction of COVID-19. Such an unprecedented challenge exposes the inadequacy in contingency planning and resources allocation strategies of public health systems. The lack of planning drives policymakers to make impromptu decisions on resources acquisitions and allocations that may have exacerbated the extraordinary shortage. 

The United States boasts one of the best healthcare systems with a large network of healthcare professionals, best-equipped medical facilities and hospitals, and most advanced medical technology. Yet, the country is under-prepared for the COVID-19 pandemic. There were severe shortages of diagnostic and preventative medical supplies both for healthcare providers and the general public in many states in the early stage of the pandemic, which made it difficult for public authorities to contain the pandemic.   According to a recent report by President Obama's former advisors on science and technology \citep{Holdren2020}, there are several contributing factors to the lack of medical resources: (1) {\it National reserves of critical medical supplies were not replenished sufficiently prior to COVID-19.} Strategic National Stockpile (SNS) was established by the government in 2003 as the national repository of pharmaceutical and vaccination stockpiles. The SNS relies on the appropriation of funding from the Congress. Much of the mask stockpile was depleted during the 2009 H1N1 pandemic and the Congress has not acted quickly to provide the funding to replenish the stockpile to an appropriate level projected by many studies. (2) {\it In order to minimize inventory cost and improve efficiency, many manufacturers and supply chain management of medical supplies have shifted to just-in-time (JIT) inventory system prior to the pandemic.} Goods are only received just in time for production and distribution. The JIT system relies the ability of manufacturers to accurately predict the demand. The initial public policies such as lock-downs caused major disruptions to supply chains around the world and there were not sufficient inventories to absorb surge demand. (3) {\it There was a lack of sufficient coordination among federal and state governments to deploy existing resources to the most devastated areas in the country.} Healthcare professionals are put at high risk to treat patients without sufficient personal protective equipment. It was difficult to uncover and contain the spread of disease without adequate testing. Several states in the United States acted on their own to secure supplies from foreign manufactures and engaged in a price bidding war for limited supplies \citep{Estes2020}. Existing resources are not necessarily distributed on a basis of health need \citep{TobinTyler2020}. Many hardest hit states have to ration care, while other states have low utilization of their resources.

While there are policy related issues that require policymakers' actions, academics can contribute to the understanding of pandemic evolution and the resulting dynamics of demand and supply. There is a clear need for the development of scientific foundation for adaptive strategies for balancing demand and supply and rationing limited resources. Contingency planning and resources allocation in a centralized form have been advocated as two coherent strategies to mitigate catastrophic economic consequences from a pandemic; see for example \citet{Jamison2017}. \citet{Ranney2020} argued that the government should have tracked the use of resources and the projection of needs in all subsidiaries, and should have coordinated allocation of resources to reduce shortage across subsidiaries and over time in the course of a pandemic. A comparable example of centralized planning is the Federal Emergency Management Agency, which administers many pre-disaster risk mitigation programs, such as national flood insurance, mitigation grants, and post-disaster response plans including search and rescue team, medical assistance teams, monetary reliefs, etc; see \citet{Vanajakumari2016} and \citet{Stauffer2020} for information. 

Centralized planning and resources allocation have long been practiced as risk management strategies in the financial industry. For example, banks and insurers are heavily regulated by governments to ensure their capabilities to absorb severe financial losses and endure adverse economic scenarios \citep{Segal2011}. The {\it central hypothesis} of this paper is that many risk aggregation and capital allocation techniques drawn from financial and insurance literature can be extended and applied to pandemic resources planning and allocation.

%


 \begin{figure}[h!]
\centering
\includegraphics[width=\textwidth]{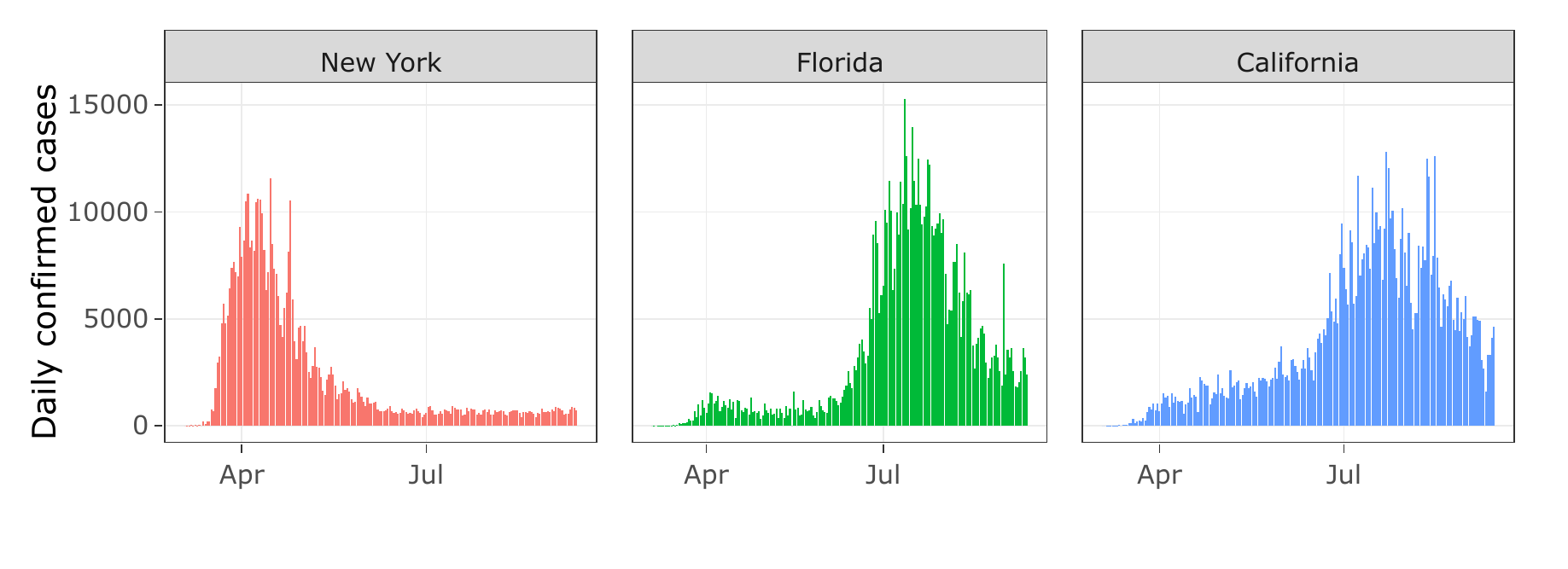}
\caption{New York, Florida, and California experience different phases of COVID-19, based on data as of September 12, 2020 from \citet{Atlantic2020}}
\label{daily_confirmed_cases}
\end{figure}

\subsection*{Case Study}\label{sec:case_study}


Figure \ref{daily_confirmed_cases} depicts the daily confirmed COVID-19 cases in the State of New York (NY), Florida (FL), and California (CA) in 2020. One could observe that, among these three states, NY was first hit the hardest by the pandemic in April, while both FL and CA experienced a peak of cases around July and August; moreover, although both NY and FL showed comparatively fewer cases in May and June, CA was savaged by the COVID-19 around that time.

Let us take a thought experiment for the moment. Imagine that these three states established resource pooling alliance prior to the pandemic. They could have complemented each other by delivering one state's surplus resources to aid another in deficit. For example, in April, the alliance could have coordinated  the efforts to send initial stockpile and increase emergency production to support the NY; in May and June, remaining resources, together with additional production, should have been redirected from both NY and FL to CA; by July and August, when both FL and CA were hit the hardest, unused resources in the NY could be made available to both FL and CA. 
Such a coalition is not unimaginable even in a decentralized political system like the United States. In April 2020, six northwestern states including New York, Connecticut, New Jersey, Rhode Island, Pennsylvania, and Delaware, formed a government procurement coalition for critical medical equipment in an effort to avoid a bidding war \citep{Holveck2020}. 

The purpose of this paper is to propose an overarching framework for different regions to optimize stockpiling and resources allocation at different pandemic stages in order to best utilize limited resources. While we use inter-state resources pooling as an illustrative example, applications can also include international collaboration on the production, procurement, distribution and pooling of critical medical resources such as masks, ventilators, pharmaceuticals, vaccines, etc.

\subsection*{Pandemic Risk Management Framework and Contribution}

A vast amount of recent literature on COVID-19 focus on the prediction of transmission dynamics (e.g. \citet{FernandezVillaverde2020, Hortacsu2020}), infected cases (e.g. \citet{Giordano_2020}), economic impact (e.g. \citet{Acemoglu2020, Gregory2020}), and the effect of non-pharmaceutical intervention and other public policies (e.g. \citet{Charpentier2020}). However, to the best of our knowledge, academic research on {\it quantitative framework} for contingency planning and resources allocation in response to pandemic risk are rare. 
While the banking and insurance industries have long had a rich tradition of developing technologies for robust risk management, the focus has been largely on financial and insurable risks.  This paper aims to take advantage of the vast medical literature on epidemic modeling and apply classic concepts from risk management and insurance literature such as reserving and capital allocation to pandemic risk management.

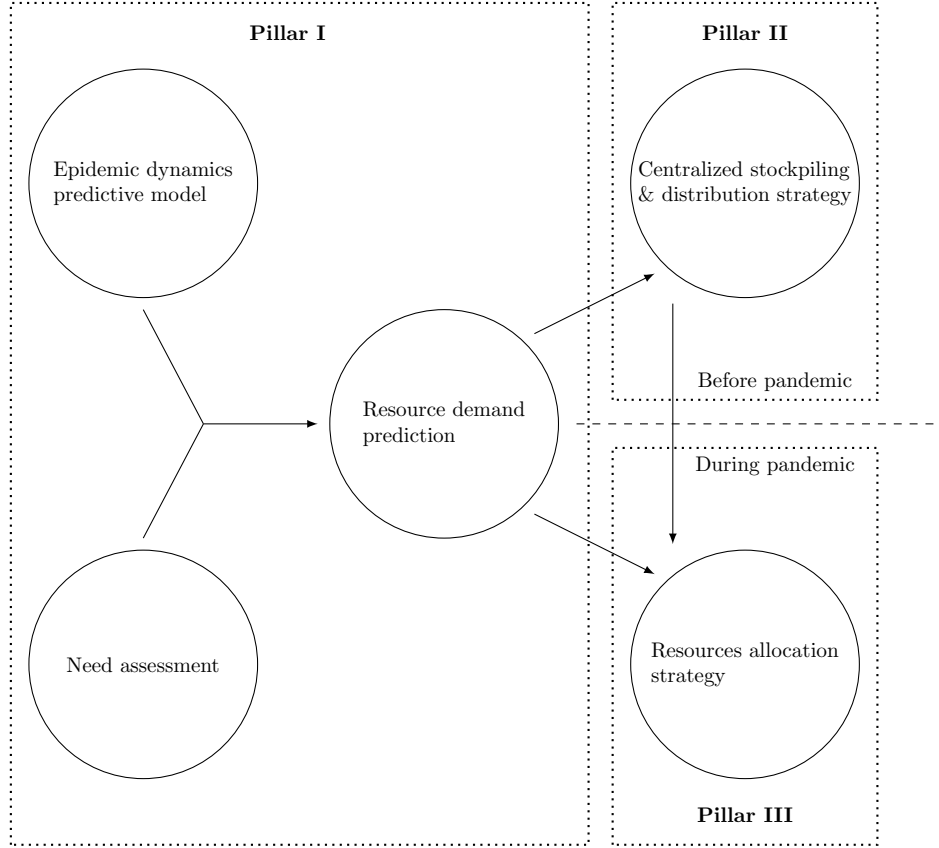
\begin{figure}[h!]
\centering
\begin{tikzpicture}[scale=0.8, transform shape]
\node[draw, align=left, circle, inner sep=0pt, minimum size=3.8cm] at (-10,4) {Epidemic dynamics\\ predictive model };
\node[draw, align=left, circle, inner sep=0pt, minimum size=3.8cm] at (-10,-4) {Need assessment};
\draw (-10,1.9) -- (-9,0);
\draw (-10,-1.9) -- (-9,0);
\draw[->, >=latex] (-9,0) -- (-7.1,0);
\node[draw, align=left, circle, inner sep=0pt, minimum size=3.8cm] at (-5,0) {Resource demand\\prediction};
\draw[->, >=latex] (-3.5,1.5) -- (-1.5,2.5);
\node[draw, align=left, circle, inner sep=0pt, minimum size=3.8cm] at (0,4) {Centralized stockpiling\\\& distribution strategy};
\draw[->, >=latex] (-3.5,-1.5) -- (-1.5,-2.5);
\node[draw, align=left, circle, inner sep=0pt, minimum size=3.8cm] at (0,-4) {Resources allocation\\ strategy};
\draw[->, >=latex] (-1.2,2) -- (-1.2,-2);
\draw[dashed] (-2.8, 0) -- (3.5, 0);
\node[align=left] at (0.5,0.7) {Before pandemic};
\node[align=left] at (0.5,-0.7) {During pandemic};
\draw[dotted, thick] (-2.2,0.4) rectangle (2.2,7);
\draw[dotted, thick] (-2.2,-7) rectangle (2.2,-0.4);
\draw[dotted, thick] (-12.2,-7) rectangle (-2.6,7);
\node[align=center] at (-7.6,6.5) {\textbf{Pillar I}};
\node[align=center] at (0,6.5) {\textbf{Pillar II}};
\node[align=center] at (0,-6.5) {\textbf{Pillar III}};
\end{tikzpicture}
\caption{Three-pillar pandemic risk management framework}
\label{fig:input_output}
\end{figure}

A strategic pandemic planning requires scientific assessment, rather than on-the-fly ad-hoc decisions and patchworks for damage control. 
In accordance with current practices of national pandemic preparedness and control strategies around the world, we summarize and propose a three-pillar framework for quantitative pandemic risk management:
\begin{itemize}
\item[] {\it Pillar I:  Regional and Aggregate Resources Supply and Demand Forecast.} Any pre-pandemic preparation plan should consist of supply and demand assessment and forecast. The supply side should include inventory assessments of critical resources and supplies, the maximum capacity of services, the capability of emergency acquisition and production. The demand side requires an understanding of the dynamics of a potential pandemic across regions and across borders. Historical data and predictive models can be used to project the evolution of a pandemic and the resulting surge demand around a healthcare system. 
\item[] {\it Pillar II: Centralized Stockpiling and Distribution}. A central authority coordinates the efforts to develop a national preparedness strategy and to set up reserves of critical resources including preventative, diagnostic and therapeutic resources. A response plan is also necessary to understand how the central authority can deliver resources to different region quickly to meet surge demands and to balance competing interests and priorities.
\item[] {\it Pillar III: Central-Regional Resources Allocation.} A pandemic response plan is critical for a central authority to contain and control the spread of a pandemic in regions under its jurisdiction. As demand may exceed any best-effort pre-pandemic projection, the authority needs to devise optimal strategies that best utilize limited existing resources and minimize the economic cost of supply-demand imbalance. A coordination strategy needs to be in place to ensure smooth communications with regional authorities. The allocation strategy should be based scientifically sound methods taking into account spatio-temporal differences across regions to ensure fairness and impartiality.
\end{itemize}

It should be pointed out that, while the first pillar is not the focus of this paper, it plays a critical role in ensuring the adequacy and effectiveness of planning and responses in second and third pillars. The proposed framework applies regardless of predictive models for projecting reported cases. 
{\it The main technical contribution of this paper lies in second and third pillars for which we propose centralized resources stockpiling, distribution, and allocation strategies.} 

\begin{table}[h!]
\begin{center}
\begin{tabular}{|l|l|}
\hline
Capital risk management & Pandemic risk management\\\hhline{|=|=|}
Business line and aggregate risk & Regional and aggregate resources demand\\\hline
Risk-based capital & Centralized stockpiling\\\hline
Business line capital allocation & Centralized distribution and allocation\\\hline
Trade-off between surplus/deficiency and cost of capital & Balance of supply/demand and economic cost \\\hline
\end{tabular}
\caption{\label{rm}Comparison between capital and pandemic risk managements}
\end{center}
\end{table}

The paper draws inspirations from two sources in insurance and risk management literature. (1) Insurance applications of epidemic models. Early applications of epidemic compartment models appeared in \citet{Hua2009}, in which real option pricing is used for operational risk management, and \citet{FenGar}, which analyzed epidemic insurance coverage. The study of epidemic insurance was extensively developed in stochastic setting in \citet{LefPic}, \citet{LefPic2018}, and \citet{LefSim}, and more recently to cyber risk assessment by \citet{HilOli}. All of these compartmental models can be used in Pillar I. (2) Capital allocation. The subject of capital management is well-studied in the insurance literature. The applications of reserving and capital allocation form the basis of the proposed Pillars II and III.
Table \ref{rm} reveals how our proposed framework shadows the classical capital management. While spatial balancing of allocation is well-known in banking and insurance (see, for instance, \citet{Dhaene2011}, and \citet{Chong2020}), this paper develops a {\it novel spatio-temporal balancing} of resources distribution and allocation, which, to the best of our knowledge, was not previously studied in either financial or management literature.

The rest of the paper is organized as follows. Each of the next three section provides detailed discussion of one of the three pillars in the proposed pandemic risk management framework as well as economic interpretations of resulting optimal strategies. Numerical examples are embedded in the discussion for better readership. We conclude in the last section with discussions of potential applications and future work.


\section*{Pillar I: Regional and Aggregate Resources Demand Forecast}\label{sec:forecast}
In the pre-pandemic time, a central authority should first model the pandemic transmission dynamics in each region. Regardless of the choice of epidemiological models, the central authority should calibrate the model in each region by its preparedness and other contingency measurements. Indeed, epidemic forecast models have been used for healthcare policy making and public communications; see, for example, \citet{Leung2020} and \citet{Tian2020}. In this paper, in line with \citet{can2020} and \citet{Hill2020}, the population in each region is divided into seven mutually exclusive compartments, namely, the susceptible ($S$), the exposed ($E$), the mild infected ($I_1$), the infected with hospitalization ($I_2$), the infected with intensive care ($I_3$), the recovered ($R$), and the deceased ($D$). The dynamics, among these seven compartments, are governed by a set of ordinary differential equations, and the model is, in short, called the SEIRD.

{\allowdisplaybreaks
This SEIRD model is characterized by a set of ordinary differential equations that describe population flows among all aforementioned compartments:
\by
\df S(t)&=&-(\beta_1 I_1(t)+\beta_2 I_2(t)+\beta_3 I_3(t)) S(t) \df t,\\
\df E(t)&=& [(\beta_1 I_1(t)+\beta_2 I_2(t)+\beta_3 I_3(t))S(t)-\gamma E(t)] \df t,\\
\df I_1(t)&=& [\gamma E(t)-(\delta_1+p_1) I_1(t)] \df t,\\
\df I_2(t)&=& [p_1I_1(t)-(\delta_2+p_2)] I_2(t)] \df t,\\
\df I_3(t)&=& [p_2I_2(t)-(\delta_3+\mu)I_3(t)] \df t,\\
\df R(t)&=&[\delta_1I_1(t)+\delta_2I_2(t)+\delta_3I_3(t)]\df t,\\
\df D(t)&=&\mu I_3(t) \df t.
\ey
All parameters in the set of equations bear clinical meanings; $\beta_i$, $i=1,2,3$, is the transmission rate of the infected class $I_i$; $1/\gamma$ is the average latency period; $1/\delta_i$, $i=1,2,3$, is the average duration of infection in the class $I_i$ before recovery to the class $R$; $p_i$, $i=1,2,3$, represents the rate at which conditions worsen and individuals require healthcare at the next level of severity; $\mu$ is the rate for most severe cases in the class $I_3$ to the deceased class $D$. Suppose that the total number of individuals in the entire popluation is $N.$ Each of the ordinary differential equation represents a decomposition of instantaneous change in the population of a compartment. For example, the first equation shows that the instantaneous rate of reduction in the number of susceptible, $- \df S(t)$ matches the sum of the rates of infection due to contacts with the infected of all classes, $\beta_1 I_1(t) S(t)+\beta_2 I_2(t) S(t)+\beta_3 I_3(t)S(t)$. The products are due to the law of mass action in biology. For example, the rate of secondary infection by the mildly infected, $(\beta_1 N) I_1(t) (S(t)/N)$ can be interpreted as the number of adequate contact each infected makes to transmit the disease $\beta_1 N$ multiplied by the number of infected $I_1(t)$, multiplied by the percentage that each contact is made with a susceptible, $S(t)/N.$ All other equations can be explained in similar ways. The estimation of these model parameters are well-studied in the literature for the COVID-19 as well as other pandemics, such as \citet{Wu2020}, \citet{Yang2020}, and \citet{Yang2020a}. Based on these parameters, the basic reproductive ratio $R_0$ of a pandemic is given by:
\begin{equation*}
R_0=\frac{N}{p_1+\delta_1}\left(\beta_1+\frac{p_1}{p_2+\delta_2}\left(\beta_2+\beta_3\frac{p_2}{\mu+\delta_3}\right)\right).
\end{equation*} The basic reproductive ratio $R_0$ can be estimated by empirical data and is often used to calibrate other parameters.
In what follows, we shall use discretized version of the compartmental model. For example, we use the notation $I_{1,j}=I_1(j \Delta t )$ to indicate the number of mild cases projected on the $j$-th period each with the length $\Delta t.$ We sometimes omit the information on $\Delta t$ as the time unit may vary depending on the reporting period.}


%

Based on predictive models such as the above-mentioned regional SEIRD models, a central authority could predict, prior to a pandemic or at the onset of a pandemic, changes in demand over the course of the pandemic. Resources require different stockpiling and allocation strategies, depending on their shelf lives. In this paper, we consider two types of medical resources, namely durable and single-use. Durable resources refer to those that can perform their required functions for a lengthy period of time without significant expenditures of maintenance or repair. Single-use resources are those that are designed to be used once and then disposed of. Mechanical ventilators and PPE are used as representative examples of durable and single-use resources respectively in this paper. \\

\begin{table}[h!]
\centering
\begin{tabular}{c|c|c}
& $\alpha\in\left[0,1\right]$ & Data source \\ \hline
China & $<20\%$ &  \citet{Yang2020a} \\  
Italy & $[87\%,90\%]$ &  \citet{Grasselli2020} \\
Seattle & $75\%$ &  \citet{Bhatraju2020}\\
Washington & $71.4\%$ &  \citet{Arentz2020}
\end{tabular}
\caption{Percentage of severe ICU infected cases requiring ventilators}
\label{tbl:vent}
\end{table}

\noindent {\bf Durable Resource: Ventilator}\\

Based on the findings in medical literature (references within Table \ref{tbl:vent}), there are estimates of the percentage $\alpha$ of the infected with intensive care that require the use of mechanical ventilators. These regional differences can be addressed in separate regional compartment models. We can use these estimates to project the ventilator demand by $X_j^{\text{VEN}\left(i\right)}=\alpha I_{3,j}^{\left(i\right)}$, where $i$ indicates the $i$-th region in the alliance and $j$ indicates the $j$-th day of the pandemic. The model can also be extended to include time-varying percentage of severe patients requiring ventilators. The calculations in the rest of the paper would carry through. Figure \ref{fig:demand}(a) shows the projected ventilator demands in New York, Florida and California based on the SEIRD model proposed by \citet{can2020}, which is calibrated to publicly available reported cases as of September 12, 2020, and demand assessment parameters in Appendix \ref{sec:para}.\\

\noindent {\bf Single-Use Resource: Personal Protective Equipment}\\

The assessment of need for PPE sets varies by the patients' class and severity of medical conditions in care, and the function of medical professionals. Table \ref{ECDCTable2} offers an example of such need assessment by \citet{ECDPC2020}. Given these estimates, we can project the regional PPE set demand by $X_j^{\text{PPE}\left(i\right)}=\theta^{E}(S^{\left(i\right)}_{j-1}-S^{\left(i\right)} _j)+\theta^{I_2}I^{\left(i\right)}_{2,j}+\theta^{I_3}I^{\left(i\right)}_{3,j}$, where $\theta^{E}$ is the number of PPE sets per exposed case, $\theta^{I_2}$ is the number of PPE sets per day per hospitalized patient, and $\theta^{I_3}$ is the number of PPE sets per day per intensive care patient. Note that $S_{j-1}-S_j$ represents the daily exposed cases whereas $I_{2,j}$ and $I_{3,j}$ keeps track of existing infected cases that require medical attention. Figure \ref{fig:demand}(b) shows how ventilator and PPE demands are projected to evolve over time for New York, Florida, and California, based on the model by \citet{can2020} and the PPE need assessment in Appendix \ref{sec:para}.

\begin{table}[h!]
\centering
\begin{tabular}{|c|c|c|c|}
\hline
& Suspected & \makecell[c]{Infected \\
Hospitalized cases} & \makecell[c]{Infected \\
ICU cases} \\
\hline  
& Number of sets per case  &  \multicolumn{2}{c|}{Number of sets per day per patient}  \\ 
Healthcare staff & $\theta^E$ & $\theta^{I_2} $ & $\theta^{I_3}$ \\
\hhline{|=|=|=|=|}
Nursing & 1-2 & 6 & 6-12 \\  
\hline  
Medical & 1& 2-3 & 3-6 \\
\hline  
Cleaning & 1 & 3 & 3 \\
\hline  
\makecell[c]{Assistant nursing \\ and other services} & 0-2 & 3 & 3 \\
\hline
Total & 3–6 & 14–15 & 15–24 \\
\hline
\end{tabular}
\caption{Minimum amount of PPE sets for different scenarios}
\label{ECDCTable2}
\end{table}

\begin{figure}[h!]
\centering
\begin{subfigure}{.5\textwidth}
\centering
\includegraphics[width=\linewidth]{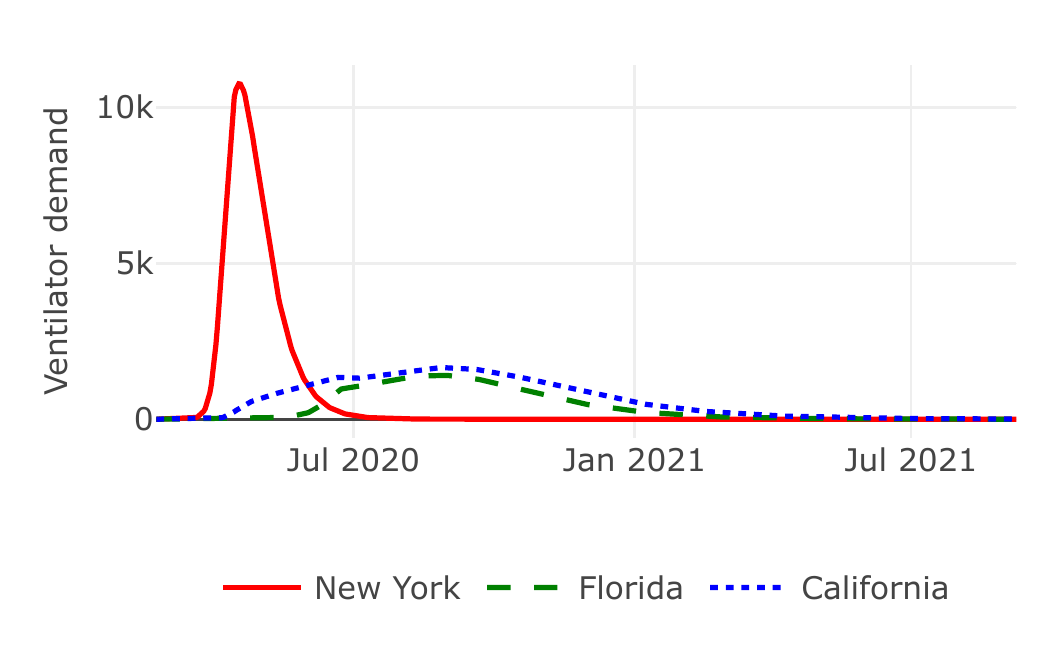}
\caption{Ventilator}
\label{fig:vent_demand}
\end{subfigure}%
\begin{subfigure}{.5\textwidth}
\centering
\includegraphics[width=\linewidth]{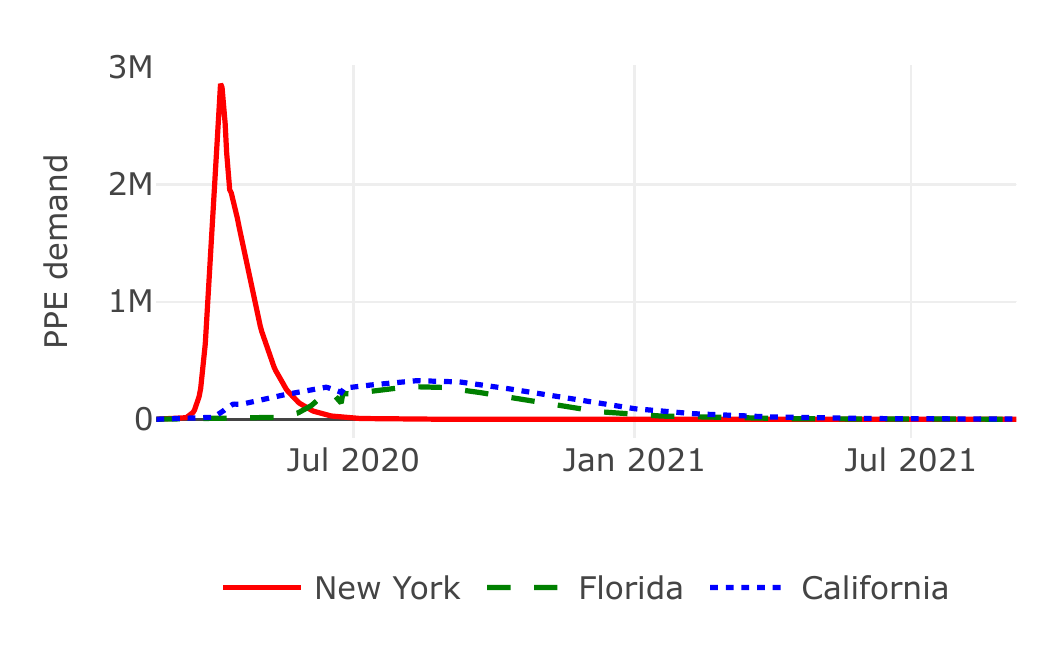}
\caption{Personal protective equipment}
\label{fig:ppe_demand}
\end{subfigure}
\caption{Ventilator and personal protective equipment regional demand prediction in New York, Florida, and California}
\label{fig:demand}
\end{figure}

In the first pillar, the central authority is expected to work with regional authorities and healthcare professionals to predict the dynamics of regional demands. All regional data are then compiled and aggregated to form the basis of forecast for the system-wide resource demand. Suppose that there are a total of $n$ regions in a healthcare system or medical resource alliance. For example, the aggregate ventilator demand can be determined by $X_j^{\text{VEN}}=\sum_{i=1}^{n}X_j^{\text{VEN}\left(i\right)}$, while the aggregate PPE set demand may be given by $X_j^{\text{PPE}}=\sum_{i=1}^{n}X_j^{\text{PPE}\left(i\right)}$. Figure \ref{fig:aggre_demand} shows how the aggregate ventilator and PPE demand prediction for the COVID-19 pandemic could have been made in the hypothetical example of a three-state resources pooling alliance. Observe that projections for ventilators and PPE sets show very similar patterns as both were driven by the same SEIRD model. The peaks in demand for ventilator are delayed compared with those for PPE sets in Figure \ref{fig:aggre_demand} due to the fact that it may take a few days before newly diagnosed patients to develop symptoms that require ventilator intervention. The projection of regional and aggregate demands offers health authorities a clear understanding of the {\it temporal competition} of critical resources.

\begin{figure}[h!]
\centering
\begin{subfigure}{.5\textwidth}
\centering
\includegraphics[width=\linewidth]{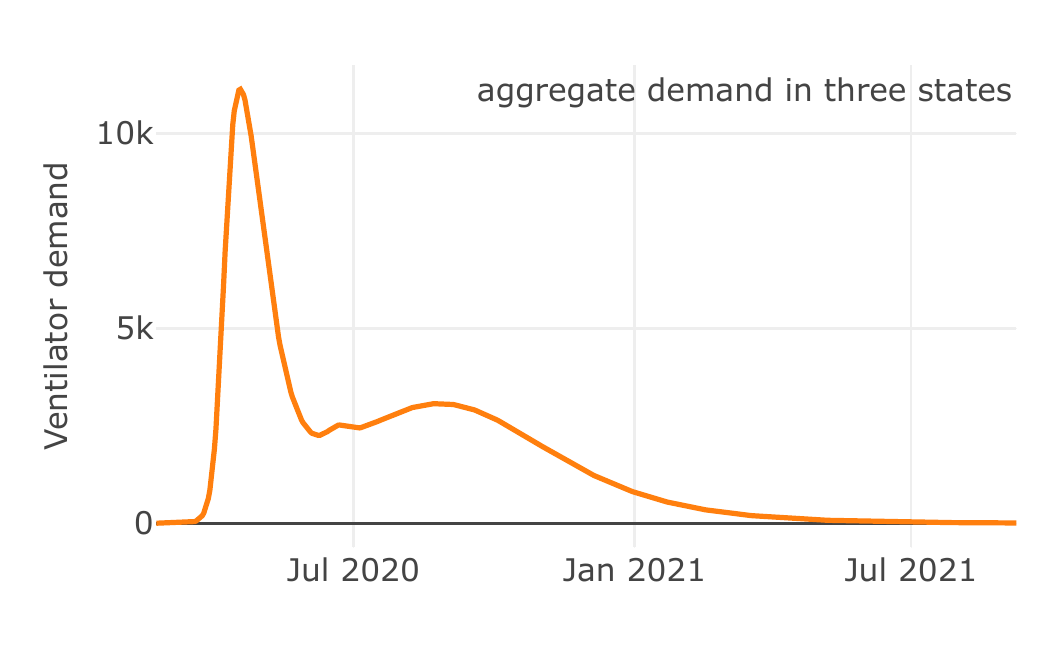}
\caption{Ventilator}
\label{fig:aggre_vent_demand}
\end{subfigure}%
\begin{subfigure}{.5\textwidth}
\centering
\includegraphics[width=\linewidth]{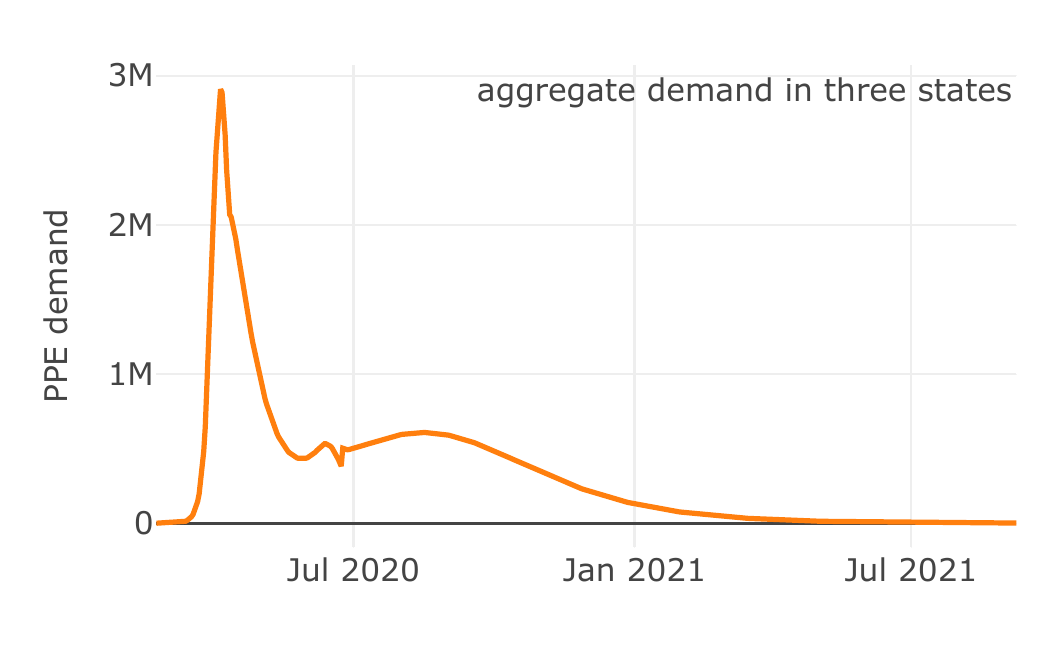}
\caption{Personal protective equipment}
\label{fig:aggre_ppe_demand}
\end{subfigure}
\caption{Ventilator and personal protective equipment aggregate demand prediction in New York, Florida, and California}
\label{fig:aggre_demand}
\end{figure}


It should be pointed out that predictive models in Pillar I, such as the SEIRD model introduced in this section, are used for multiple purposes as shown in Figure \ref{fig:input_output}. First, they need to be developed prior to a pandemic using historical data and to form the basis of demand forecast for contingency planning in Pillar II. Then, as a pandemic starts to emerge, the predictive models also need to be re-calibrated and updated with latest medical knowledge and reported cases. New forecasts would then be fed into models to determine optimal allocation strategies in Pillar III. As medical knowledge of the viral disease evolves and predictive models improve over time, Pillars I and III may be revisited from period to period. When a distribution schedule of resources requires an update, we can go back to Pillar II. Therefore, the three-pillar framework may be utilized in circles such as Pillars I, II, III, I, III, I, III, I, II, III, etc.

\section*{Pillar II: Centralized Stockpiling and Distribution}\label{sec:stock_dis}

As the pandemic unfolds, many hospitals and healthcare facilities may run out of pharmaceuticals and other essential resources before emergency production can pick up and additional supplies become available. To meet the surge demand at the onset of a pandemic, many countries maintain national repositories of antibiotics, vaccines, chemical antidotes, antitoxins, and other critical medical supplies. A {\it centralized stockpiling strategy} is intended to provide a stopgap measure to meet the surge in resources demand at the early stage of the pandemic. There has been well-established literature on stockpiling strategies for influenza pandemics; see, for example, \citet{Greer2013} and \citet{Siddiqui2008}.

One should keep in mind that a practical stockpiling strategy is often an act of {\it balance between adequate supply and economic cost}. On one hand, under-stocking is a common issue as resources and their storage can post heavy cost, and the actual demand during the pandemic outbreak could deviate from the projection; for example, \citet{UW} claimed that as many as $20$ states in the United States are expected to encounter shortage in ICU beds when the COVID-19 cases peak. On the other hand, excessive stockpiling for long term could lead to unnecessary waste, especially for disposable and perishable resources; for instance, \citet{stat} reported that, in March 2020 during the COVID-19 pandemic, the Strategic National Stockpile in the United States stocked $13$ million N95 masks, of which as many as $5$ million may have expired, partly contributing to the nation-wide shortage of masks.

In the second pillar of our proposed framework, based on the estimated aggregate resources demand, the central authority could then develop stockpiling and distribution strategies in regular time before a pandemic. Notice that durable resources such as ventilators can be reused throughout the pandemic, while single-use resources such as PPE sets must be disposed of after one-time usage. Hence, we have to treat them separately for optimal centralized stockpiling and distribution strategies.

\subsection*{Durable Resource: Ventilator}
It is typical that a central authority has to determine an optimal initial stockpile size $K_0$ of resources to maintain in some centralized location. In addition, to meet surge demand, the authority may need to reach contractual agreements with suppliers for emergency orders, which may be limited by the maximum production rate of $a$ units per day during a pandemic. Since ventilators are durable, the stock of ventilators does not decrease over time due to usage. We assume that they can be deployed to different regions at negligible cost.  Therefore, the total number of available ventilators in the entire alliance is given by $K_j=K_0+aj$, on the $j$-th day since the onset of the pandemic. Hence, the only decision variable of the central authority in the case of ventilators is the initial stockpile size $K_0$.

\begin{figure}[h!]
\centering
\begin{subfigure}{.5\textwidth}
\centering
\includegraphics[width=\linewidth]{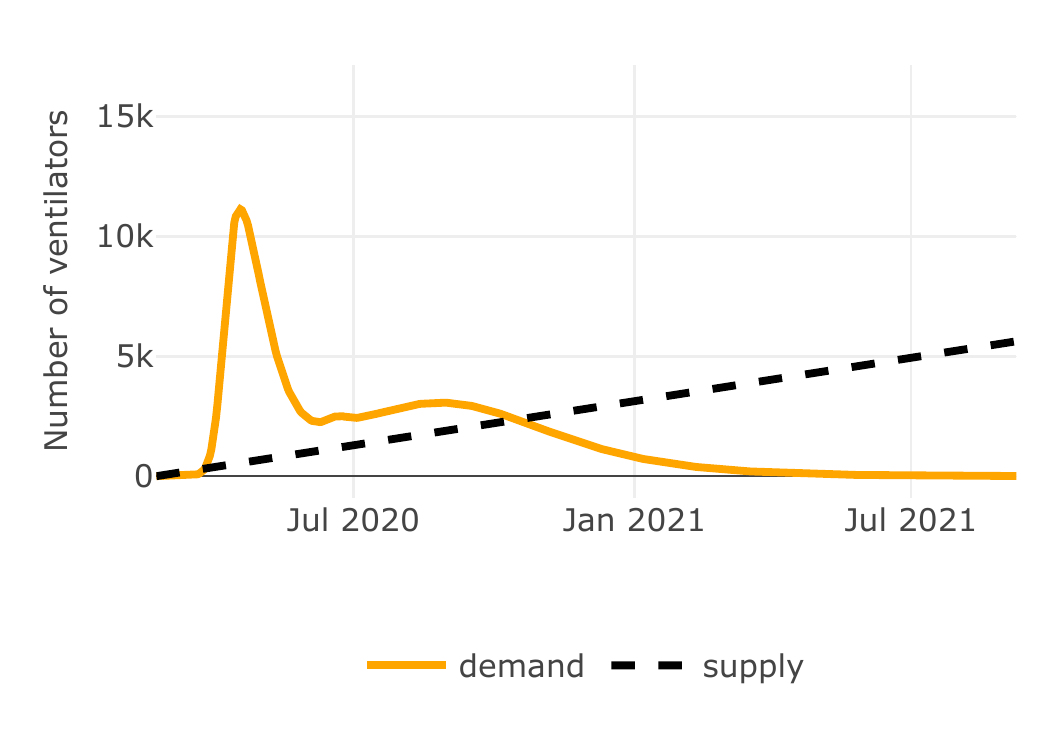}
\caption{Extreme understocking}
\label{fig:aggre_vent_supply_lo}
\end{subfigure}%
\begin{subfigure}{.5\textwidth}
\centering
\includegraphics[width=\linewidth]{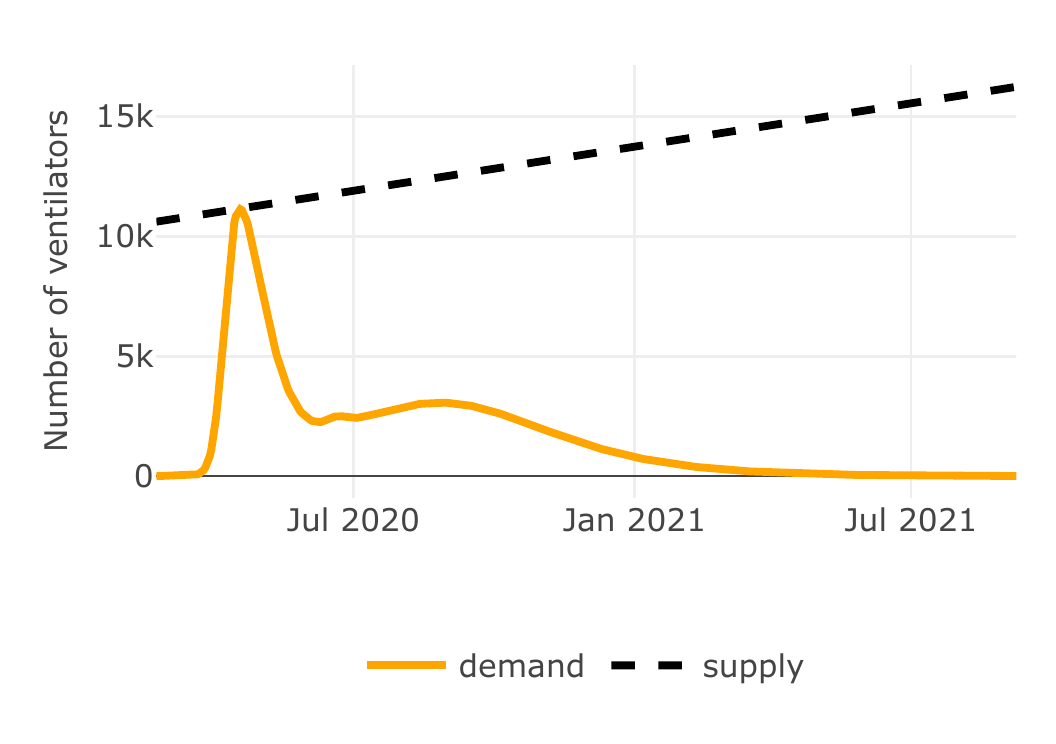}
\caption{Extreme oversupply}
\label{fig:aggre_vent_supply_hi}
\end{subfigure}
\caption{Two extreme scenarios of initial stockpile size $K_0$ for ventilators}
\label{fig:aggre_vent_supply}
\end{figure}


To better explain the need for an optimal initial stockpile size $K_0$, consider two extreme cases in Figures \ref{fig:aggre_vent_supply_lo} and \ref{fig:aggre_vent_supply_hi} for the three-state alliance. On one end of the extreme, the central authority may decide not to hold any initial stockpile but simply rely on the maximum emergency production limit during the pandemic; Figure \ref{fig:aggre_vent_supply_lo} shows a clear extreme shortage at the first peak time of aggregate ventilator demand. On the other end, suppose that the central authority decides to hold an extraordinary amount of initial stockpile for ventilators to meet the highest peak of aggregate ventilator demand; Figure \ref{fig:aggre_vent_supply_hi} illustrates a clear extreme oversupply of ventilators during most of the time of pandemic; also, in this case, the economic cost of severe storage can be huge. Therefore, the central authority has to take a delicate balance on an initial stockpile size $K_0$ that takes into account the economic cost of shortage, oversupply, as well as storage costs. 

Consider the following optimization model for an initial stockpiling size.

\begin{equation} \label{vensum}
\min_{K_0\geq 0}\sum_{j=1}^{m}\omega_j\left(\frac{\theta^{+} _j}{2}\left(X^{\text{VEN}} _j-\left(K_0+aj\right)\right)_{+}^2+\frac{\theta^{-} _j}{2}\left(X^{\text{VEN}} _j-\left(K_0+aj\right)\right)_{-}^2+c_j\left(K_0+aj\right)\right)+c_0K_0,
\end{equation}
where $m$ is the number of days of the pandemic, $\omega_j$ is a weight for significance of precision for the costs on the $j$-th day of the pandemic, $\theta_j^{+}$ is an economic cost per squared unit of shortage, $\theta_j^{-}$ is an opportunity cost per squared unit of oversupply, $c_j$ is the aggregate cost of possession per unit of ventilators per day, $c_0$ is the initial stockpile cost, which may include both the acquisition cost and expected cost of possession (storage, maintenance, inventory logistics, opportunity cost). The quadratic form can be interpreted as follows. While one copy of the quantity $X^{\text{VEN}} _j-\left(K_0+aj\right)$ represents the amount of resource imbalance (shortage or surplus), the other copy $(\theta^\pm_j/2) [X^{\text{VEN}} _j-\left(K_0+aj\right)]_{\pm}$ can be viewed as the (linear) variable cost of the imbalance. In other words, the larger the imbalance, the higher the price to pay. The quadratic form is the product of cost per unit and the unit of imbalance, which yields the overall economic cost of imbalance.  The weight $w_j$ can be used for different purposes. For example, it may be reasonable to assume that the precision of meeting demands in near future is more important than that in the far-future given the uncertainty with prediction. Another case may be to make the weight proportional to the daily demand $X_j$ as the demand-supply imbalance can have a greater impact on population dense areas than otherwise. 

\begin{figure}[ht!]
    \begin{center}
    \scalebox{0.9}{
    \begin{tikzpicture}
        \draw (0,-5) -- (0,3.2);
        \draw (0,-5) -- (14,-5);
        \draw (0,-5) rectangle (1.4,-4);
        \node[align = center] at (0.7, -4.3) {$Y_{[1]}$};
        \draw (1.4,-5) rectangle (2.8,-3.5);
        \node[align = center] at (2.1, -3.8) {$Y_{[2]}$};
        \node[align = center] at (3.5, -4) {$\dots$};
        \draw (4.2,-5) rectangle (5.6,-1.6);
        \node[align = center] at (4.9, -1.9) {$Y_{[J-2]}$};
        \draw (5.6,-5) rectangle (7,-1);
        \node[align = center] at (6.3, -1.3) {$Y_{[J-1]}$};
        \draw (7,-5) rectangle (8.4,1);
        \node[align = center] at (7.7, 1.3) {$Y_{[J]}$};
        \draw (8.4,-5) rectangle (9.8,1.3);
        \node[align = center] at (9.1, 1.6) {$Y_{[J+1]}$};
        \node[align = center] at (10.5, -4) {$\dots$};
        \draw (11.2,-5) rectangle (12.6,2);
        \node[align = center] at (11.9, 2.3) {$Y_{[m-1]}$};
        \draw (12.6,-5) rectangle (14, 2.6);
        \node[align = center] at (13.3, 2.9) {$Y_{[m]}$};

        \draw[dashed, thick] (-0.5, 0) -- (14.5, 0); 
        \node[align = center] at (-0.25, 0.2) {$K_0$};

        \draw[arrows = {triangle 90 - triangle 90}] (0.7, -3.9) -- (0.7, -0.1);
        \draw[arrows = {triangle 90 - triangle 90}] (2.1, -3.4) -- (2.1, -0.1);
        \draw[arrows = {triangle 90 - triangle 90}] (4.9, -1.5) -- (4.9, -0.1);
        \draw[arrows = {triangle 90 - triangle 90}] (6.3, -0.9) -- (6.3, -0.1);
        \draw[arrows = {triangle 90 - triangle 90}] (7.7, 0.9) -- (7.7, 0.1);
        \draw[arrows = {triangle 90 - triangle 90}] (9.1, 1.2) -- (9.1, 0.1);
        \draw[arrows = {triangle 90 - triangle 90}] (11.9, 1.9) -- (11.9, 0.1);
        \draw[arrows = {triangle 90 - triangle 90}] (13.3, 2.5) -- (13.3, 0.1);
    \end{tikzpicture}}
    \end{center}
    \caption{Optimal initial stockpile $K_0$ relative to projected shortages without initial stockpile}
    \label{fig:optim_K0}
    \end{figure}
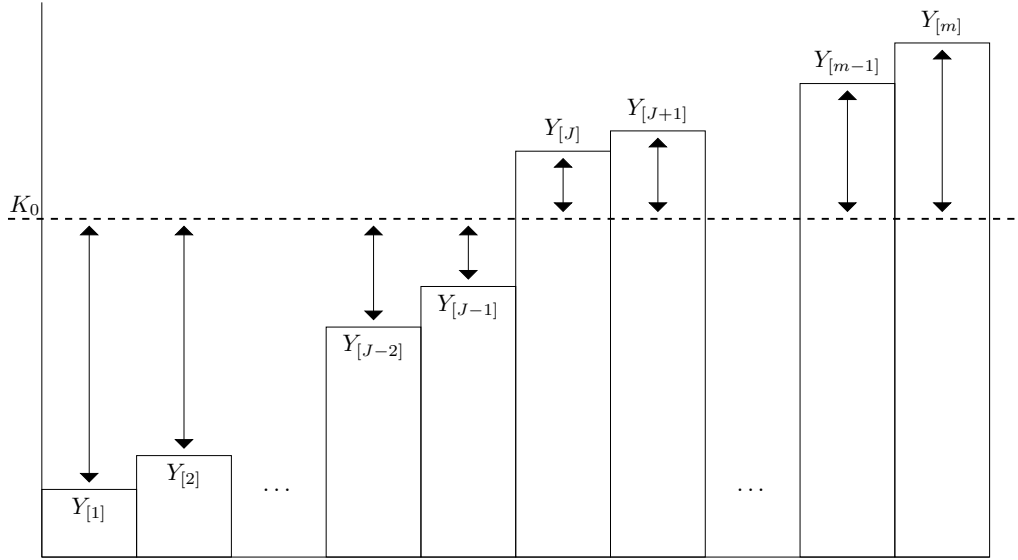

To understand the analytical solution to this problem, we need to look at the projected shortage without any initial stockpile, $Y_j:=X_j^\text{VEN} - aj$, for $j=1, \cdots, m,$ which is the accumulated demand less the accumulated supply apart from the initial stockpile. Note that we consider the accumulated supply because the resources are durable and can be reused. When $Y_j>0,$ there is a drain on the initial stockpile as current demand exceeds the accumulated supply. Otherwise, the stockpile increases as supply exceeds demand. Because the economic costs of shortage and surplus are weighed differently, the value of this objective function depends on the number of days with decreasing stockpile $(Y_j>0)$ and those with increasing stockpile $(Y_j<0)$. The analytical solution to this problem requires {\it sorting} projected shortages in an ascending order. Let us denote the sorted sequence by $\{Y_{[j]}, j=1, \cdots, m\},$ where $Y_{[j]}$ represents the $j$-th smallest projected shortage.   In the objective function \eqref{vensum}, the cost coefficient $\theta^\pm_j$ applies according to whether or not stockpile exceeds demand. If $K_0$ is placed below $Y_{[j]}$, there is a shortage in the healthcare system and hence the cost coefficient $\theta^+$ is applied. Otherwise, there is a surplus in the system and the cost efficient $\theta^-$ is applied. The optimality is achieved when $K_0$ is kept at a delicate position. See Figure \ref{fig:optim_K0}. The nature of the sum of squared shortages in \eqref{vensum} determines that the optimal initial stockpile $K^\ast_0$ should be squeezed between $Y_{[J-1]}$ and $Y_{J}$ in such a way that
\[Y_{[J-1]}\le \frac{\sum_{j=1}^{J-1}\omega_{[j]}{\theta_{[j]}^-}\left(Y_{[j]}-\frac{c_{[j]}}{\theta_{[j]}^-}\right)+\sum_{j=J}^{m}\omega_{[j]}{\theta_{[j]}^+}\left(Y_{[j]}-\frac{c_{[j]}}{\theta_{[j]}^+}\right)-c_0}{\sum_{j=1}^{J-1}\omega_{[j]}{\theta_{[j]}^-}+\sum_{j=J}^{m}\omega_{[j]}{\theta_{[j]}^+}} \le Y_{[J]}.\] Once $J$ is identified, the optimal stockpile $K_0^*$ is given by 
\begin{equation*}
    K_0^* = \max\left\{\frac{\sum_{j=1}^{J-1}\omega_{[j]}{\theta_{[j]}^-}\left(Y_{[j]}-\frac{c_{[j]}}{\theta_{[j]}^-}\right)+\sum_{j=J}^{m}\omega_{[j]}{\theta_{[j]}^+}\left(Y_{[j]}-\frac{c_{[j]}}{\theta_{[j]}^+}\right)-c_0}{\sum_{j=1}^{J-1}\omega_{[j]}{\theta_{[j]}^-}+\sum_{j=J}^{m}\omega_{[j]}{\theta_{[j]}^+}},\;\; 0\right\}.
\end{equation*}
The proof of this result can be found in Appendix \ref{appendix_b1}.  This result shows that the optimal initial stockpile $K_0$ is the weighted average of all projected shortages discounted by the cost of possession relative to the economic cost of shortage, $Y_{[j]} - c_{[j]}/\theta^{\pm}_{[j]}$. The adjustment term $c_{[j]}/\theta^{\pm}_{[j]}$ indicates that, the higher cost of possession relative to economic cost of imbalance, the fewer ventilators should be acquired. It is logical that, if the cost of possession for durable resource is too high, the central authority in a poor country may have little financial means to pay for stockpiling and be left with no choice but to deal with the demand-supply imbalance. In contrast, if the economic cost of imbalance is too high due to lost productivity or even the society's resentment on government's failure to meet demand, then the central authority would ignore the cost of possession and do everything possible to reduce the shortage. 

\begin{figure}[htb]
\centering
\begin{subfigure}{.5\textwidth}
\centering
\includegraphics[width=\linewidth]{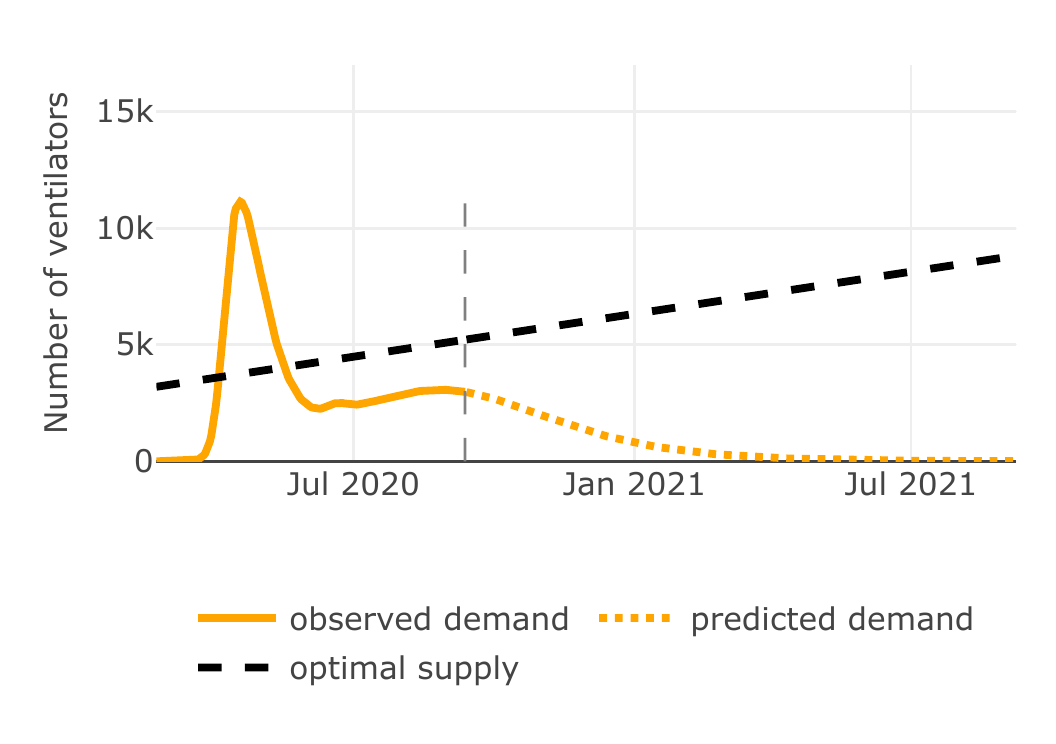}
\caption{$\theta_j^+ = \theta_j^-$ (Shortage costs the same as surplus)}
\label{fig:aggre_vent_supply_optim}
\end{subfigure}%
\begin{subfigure}{.5\textwidth}
\centering
\includegraphics[width=\linewidth]{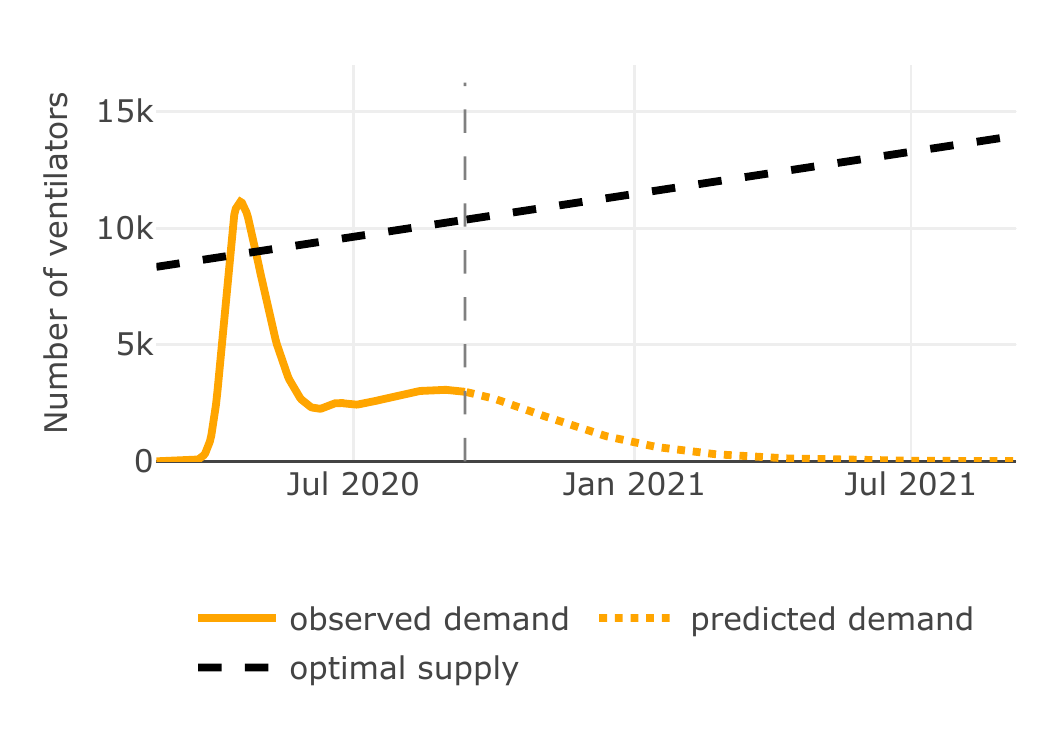}
\caption{$\theta_j^+ = 20\theta_j^-$ (Shortage costs more)}
\label{fig:aggre_vent_supply_optim_asym}
\end{subfigure}
\caption{Optimal initial stockpile size $K_0$ for ventilators according to different weights of economic cost}
\label{fig:aggre_vent_supply2}
\end{figure}

Figure \ref{fig:aggre_vent_supply2} depicts optimal initial stockpile size in the case study. The model parameters are provided in Appendix \ref{sec:para}. Observe that optimal initial stockpiles are chosen to reduce shortage in the early stages and oversupply in the late stages of the pandemic, compared with those strategies shown in Figure \ref{fig:aggre_vent_supply}. When the resource shortage costs the same or less than the resource surplus, Figure \ref{fig:aggre_vent_supply_optim} shows that the strategy requires less initial stockpile due to the excessive amount of supply after the pandemic dies down. In contrast, if the shortage weighs more than the surplus, the strategy is to reduce shortage in early stages at the expense of increasing oversupply in late stages; see Figure \ref{fig:aggre_vent_supply_optim_asym}.



\subsection*{Single-Use Resource: Personal Protective Equipment}
Similar to the case of durable resources, the central authority needs to set up an initial stockpile size $K_0$ of single-use resources such as PPE and make contractual agreements with emergency suppliers, which can provide additional supply at the production rate $a$ units per day. Since PPE is of single-use, during the pandemic, the central authority has to stockpile PPE sets not only for the present but also to potentially deploy them for later time in order to meet surge demand. Therefore, the central authority needs to plan for both initial stockpile size $K_0$ and the amount of distribution $k_j$ to all regions on day $j$. The dynamics of the centralized storage  $\{K_j, j=0, 1, \dots, m\}$ is determined by the recursive relation $K_j=K_{j-1}+a-k_j+\left(k_j-X_j\right)_+$ for $j=1, \dots, m$. The relation can be interpreted as follows. The current stockpile $K_j$ is based off the previous period's stockpile $K_{j-1}$, increased by the net surplus of new supply $a$ less the arranged distribution up to the total demand, $\max\{k_j, X_j\}.$ In other words, if we arrange to distribute $k_j$ units but can only consume $X_j<k_j,$ then the unused amount should count towards the centralized storage for future use.

\begin{figure}[h!]
\centering
\begin{subfigure}{0.5\textwidth}
\centering
\includegraphics[width=\linewidth]{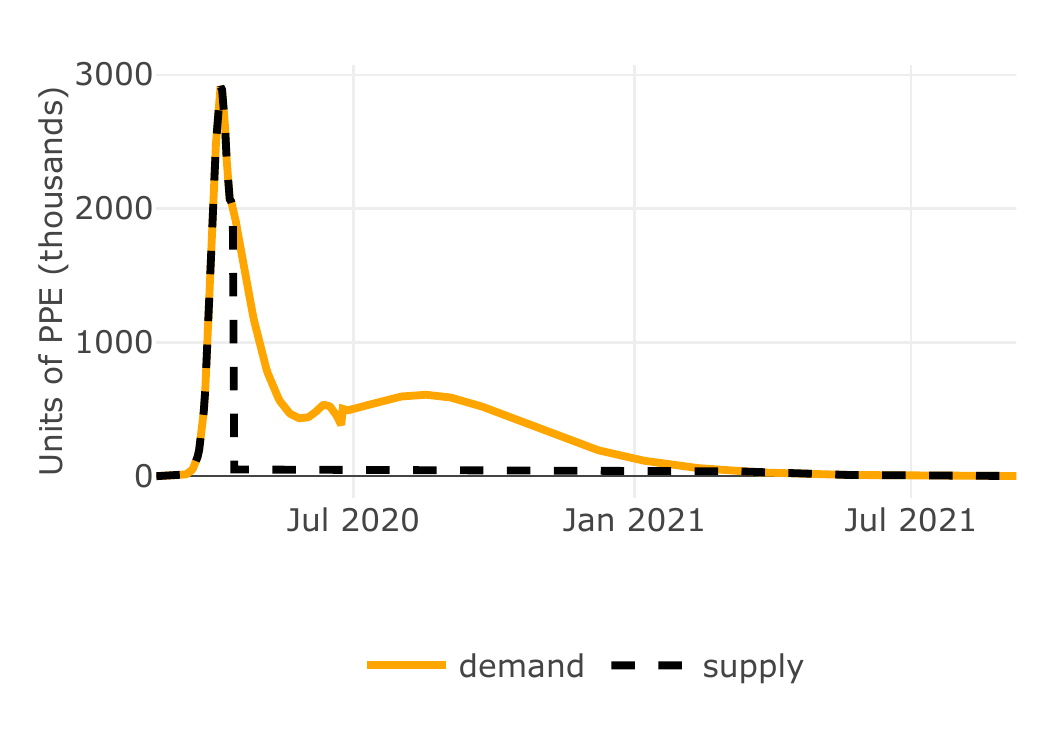}
\caption{Aggressive distribution}
\label{fig:aggre_ppe_supply_lo}
\end{subfigure}%
\begin{subfigure}{0.5\textwidth}
\centering
\includegraphics[width=\linewidth]{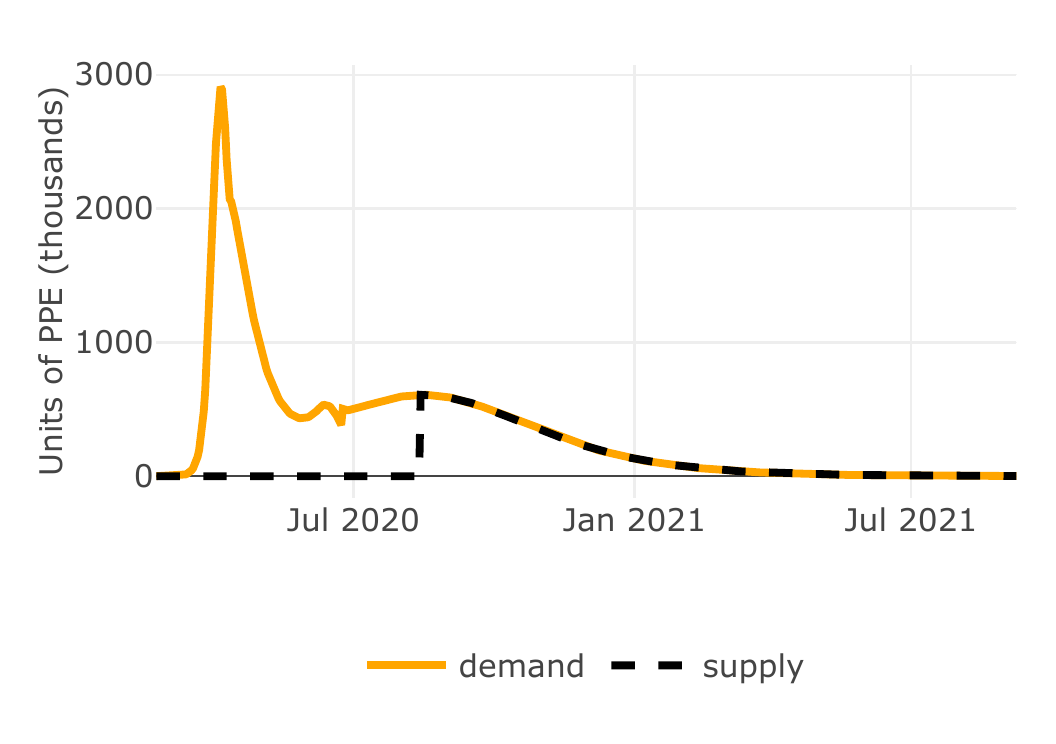}
\caption{Conservative distribution}
\label{fig:aggre_ppe_supply_hi}
\end{subfigure}
\caption{Two extreme scenarios of distribution schedule $k_1,k_2,\dots,k_m$ for personal protective equipment}
\label{fig:aggre_ppe_supply}
\end{figure}


\begin{figure}[h!]
\centering
\includegraphics[width=.5\linewidth]{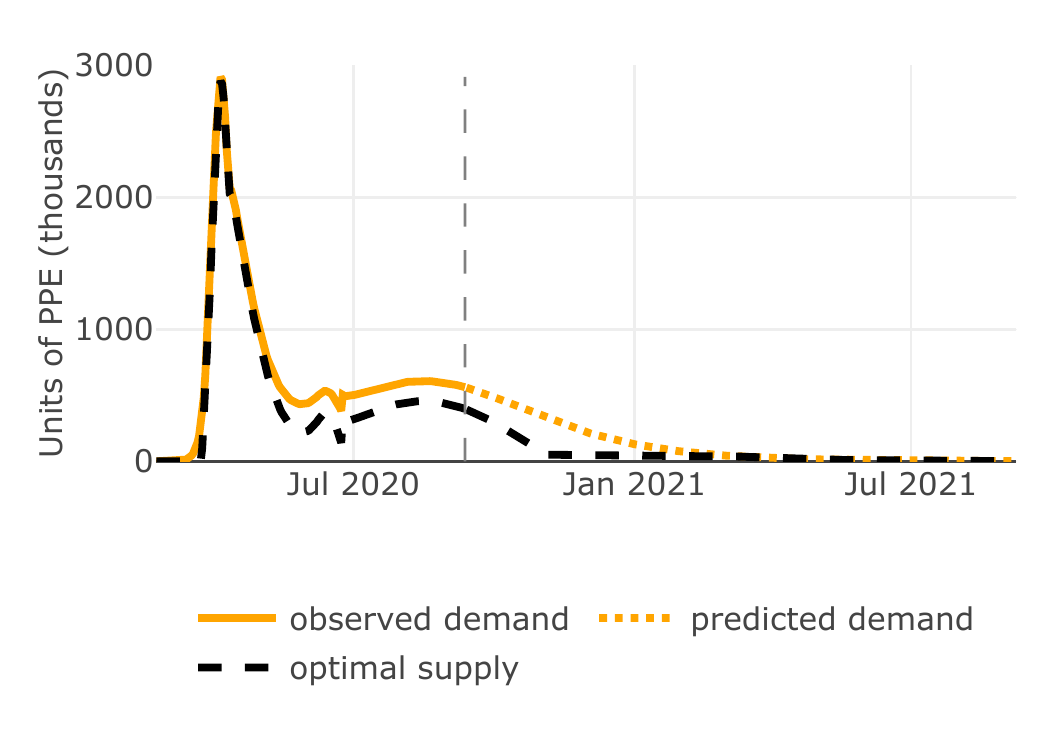}
\caption{Optimal distribution schedule $k_1,k_2,\dots,k_m$ and initial stockpile size $K_0$ of personal protective equipment for the three-state resource pooling alliance}
\label{fig:aggre_ppe_supply_optim}
\end{figure}

Consider two extreme cases in Figures \ref{fig:aggre_ppe_supply_lo} and \ref{fig:aggre_ppe_supply_hi} in the case study. In an extreme case, assume that the central authority decides to distribute as much as possible to meet the demand until the centralized storage is exhausted. This aggressive early distribution strategy is depicted in Figure \ref{fig:aggre_ppe_supply_lo}. After the storage depletion, the system relies only on the new supply, which clearly is not sufficient to meet demands and can cause severe shortage at the time of second peak. In the other extreme, the central authority may choose to hold off dispersing any equipment at all till the point that the storage is believed to be sufficient to cover all future demands. Such a conservative distribution strategy is illustrated in Figure \ref{fig:aggre_ppe_supply_hi}. The challenge with this strategy is that the central authority would have to deal with the repercussion of not providing any assistance in the early stage of the pandemic. Therefore, it is sensible that the central authority develop a distribution schedule that takes a temporal balance of varying needs from all regions. Here we introduce the optimization problem for both an initial stockpile and the distribution schedule of single-use resources.
\begin{align*}
\min_{K_0\geq 0,k_1,\dots,k_m}&\;\sum_{j=1}^{m}\omega_j\left(\frac{\theta^{+} _j}{2}\left(X^{\text{PPE}} _j-k_j\right)_+^2+\frac{\theta^{-} _j}{2}\left(X^{\text{PPE}} _j-k_j\right)_-^2+c_jK_j\right)+c_0K_0\\\text{such that}&\;K_j=K_{j-1}+a-k_j+\left(k_j-X^{\text{PPE}} _j\right)_+\geq 0\text{ and }k_j\geq 0,\quad\text{for }j=1,2,\dots,m,
\end{align*}
where $c_j$ is the centralized cost of possession per unit of PPE per day. It should be pointed out that the centralized storage should be kept non-negative for practical purposes and the distribution amount should also be kept non-negative. Negative distribution could mean confiscation of regional resources for system wide re-distribution, which is not considered in this paper. 

Figure \ref{fig:aggre_ppe_supply_optim} depicts the case of optimal distribution schedule for the three-state alliance. Observe that the optimal supply distribution schedule stays below the trajectory of demand. The slight shortage results from the consideration of the cost of possession. Should the cost of possession be zero, the optimal supply would be to match the demand exactly at all times. The initial stockpile can be set artificially high so that any desired amount can be carried over from period to period and last long enough to support all future demands. In the presence of possession cost, Figure \ref{fig:aggre_ppe_supply_optim} also reveals a distribution strategy that in essence ignores the demands at the start of a pandemic and after the pandemic dies down and instead focuses on meeting demands at the first peak of the pandemic. Keep in mind that the weight of significance $\omega_j$ in this example is set to be proportional to the size of demand. Therefore, the strategy prioritizes meeting the demand in the first peak over other periods due to its high demand. If we were to choose the same weight $\omega_j$ for all periods, the shortage would be more balanced among all periods.

This optimization problem is cumbersome to be solved analytically. However, it is straightforward to show that any oversupply distribution scheme $k_j>X^{\text{PPE}}_j$ must be sub-optimal, and hence the problem can be simplified as follows.
\begin{align*}
\min_{K_0\geq 0,k_1,\dots,k_m}&\;\sum_{j=1}^{m}\omega_j\left(\frac{\theta^{+} _j}{2}\left(X^{\text{PPE}} _j-k_j\right)^2+c_jK_j\right)+c_0K_0\\\text{such that}&\;K_j=K_{j-1}+a-k_j\geq 0\text{ and }0\leq k_j\leq X^{\text{PPE}} _j,\quad\text{for }j=1,2,\dots,m,
\end{align*}
Because of the convexity, it can be solved numerically using Disciplined Convex Programming (DCP), which requires minimal computational time \citep{Grant2006}. The solution shown in Figure \ref{fig:aggre_ppe_supply_optim} is obtained with the help of an R package \code{CVXR} developed based on the DCP method \citep{Fu2020}. 

\section*{Pillar III: Centralized Resources Allocation}\label{sec:allocate}

In the time of severe resources shortage, a coordinated effort becomes necessary to obtain additional supplies and to ration limited existing resources. Existing resources are not necessarily distributed on a basis of health need or justice \citep{TobinTyler2020}. Many hardest hit states have to ration care, while other states have low utilization of their resources. As alluded to earlier, not all regions experienced surge in demands at the same time (see Figure \ref{daily_confirmed_cases}), it has been long argued that United States federal government should have tracked the current use and the projection of needs in all states and coordinated allocation of resources to reduce shortage across regions and over time, during the COVID-19 pandemic \citep{Ranney2020}.

There are two common types of resources allocation problems in the course of a pandemic, both of which can be formulated and cast in the Pillar III of our proposed framework. 

\begin{enumerate}
\item {\it Macro level resources pooling.} A central authority acts in the best interest of a union of many regions to increase supply as well as to coordinate the distribution of existing and additional resources among different regional healthcare providers. 

\item {\it Micro level rationing.} Facing an imbalance of demands and supplies in medical equipment and resources, hospitals often have to make difficult but necessary decisions to ration limited existing resources as well as new supplies.
\end{enumerate}

While in both cases the aim of allocation is to deliver limited resources to where they are needed the most, the macro level pooling large addresses spatio-temporal differences and the micro level rationing focuses on healthcare effectiveness and justice.
The setting of standards, protocols and policies can have profound impact on the functioning of a healthcare system at the time of crisis. Therefore, the best practice of resources allocation should be based on scientific assessment and evaluation rather than on-the-fly ad-hoc decisions and a patchwork of damage-control rules. 

\begin{enumerate}
\item {\it Resources allocation should be based on a holistic approach to address concerns of all stakeholders.}

There are often conflicting interests and priorities for using limited resources. For example, when medical supplies are scarce, many countries, states and cities are competing for resources. While each state acts in its best interest to acquire medical devises and protective equipment, a federal government may see the urgent need to seize control of the cargo to boost a centralized stockpile. A holistic approach aims to strike a balance among different objectives for various stakeholders.

\item {\it  Scientific methods for resources allocation should be developed under a set of optimization objectives, meet certain ethical and humanitarian criteria, take into account logistic and budgetary constraints. }

When a pandemic breaks out, it often spreads from one cluster to another in geographic areas due to its transmission dynamics and affects different sectors of a healthcare system in a chain reaction. Medical needs can vary greatly by demographics and other socio-economics factors. While there is no universal ``one-size-fits-all" solution for allocation problems, there are a set of quantifiable and justifiable criteria. While it is difficult address all of these criteria in a single model, we believe that they can be formulated similarly as in this section.

\begin{itemize}
\item {\bf Minimization of shortage and oversupply} 

Decision makers need to take into account spatio-temporal differences in demand and supply over the course of a pandemic. It is imperative for authorities to allocate more resources to epicenters of a pandemic than other regions under less imminent threat. For example, New York City was the first in the State of New York to witness the COVID-19 pandemic, when other counties had little to no reported cases \citep{Blo}. The State governor issued an executive order to take ventilators and other protective gears from underutilized private hospitals and companies. As infected cases stabilize or even decline in pandemic-savaged regions, a central authority may need to shift its attention to other areas of potential outbreaks and allocate resources in anticipation of new waves. This was also evident when many states in the United States in early stages of outbreak during the COVID-19 took preemptive measures to procure medical supplies from countries like China and South Korea which have developed production capacities after the local epidemics are under control. Therefore, it is sensible to develop an allocation strategy that minimizes shortage and oversupply across different regions and over the life-cycle of a pandemic.

\item {\bf Promoting and rewarding instrumental value}

Critical preventive gears and medical care should be provided first to healthcare workers in the front line, employees in essential businesses and critical infrastructures. Not only because they are at high risk due to their exposure to infectious disease, but also the society bears heavy economic cost when these workers fall ill and are unable to return to work. The lack of sufficient front-line workforce may cause severe disruptions to public services, which can have rippling effect on the rest of the economy. Priority access to medical care can be a critical incentive for retention.

\item {\bf Prioritizing the worst off}

The ultimate goal of a healthcare system is to save lives. Access to critical medicare treatment should be reserved for patients facing life-threatening conditions, when there is an insufficient supply of equipment such as ventilators. 

\item {\bf Maximization of benefit from treatment}

Maximization of the benefit requires prognosis on how patients are likely to survive with treatment. Some recent study of COVID-19 patients in the United States finds that most patients do not survive after being placed on mechanical ventilators \citep{Pre}. To maximize the benefit, access to ventilator treatment should be prioritized for younger patients who  can benefit the most and have the higher chance of survival. For example, many hospitals in Italy lowered the age cutoff from 80 to 75 in order to ration limited ventilators \citep{Rosenbaum2020}. Such a strategy often leads to ethical dilemma when in conflict with prioritizing the worst off.

\end{itemize}

\end{enumerate}

The third pillar of the proposed pandemic risk management framework is to allocate limited resources for different regions, building on the proposed optimal centralized stockpiling and distribution strategies. 
Figures \ref{fig:vent_allo_demands} and \ref{fig:ppe_allo_demands} put the regional resources demand and optimal aggregate supply together for the ease of exposition.

\begin{figure}[h!]
\centering
\begin{subfigure}{.5\textwidth}
\centering
\includegraphics[width=\linewidth]{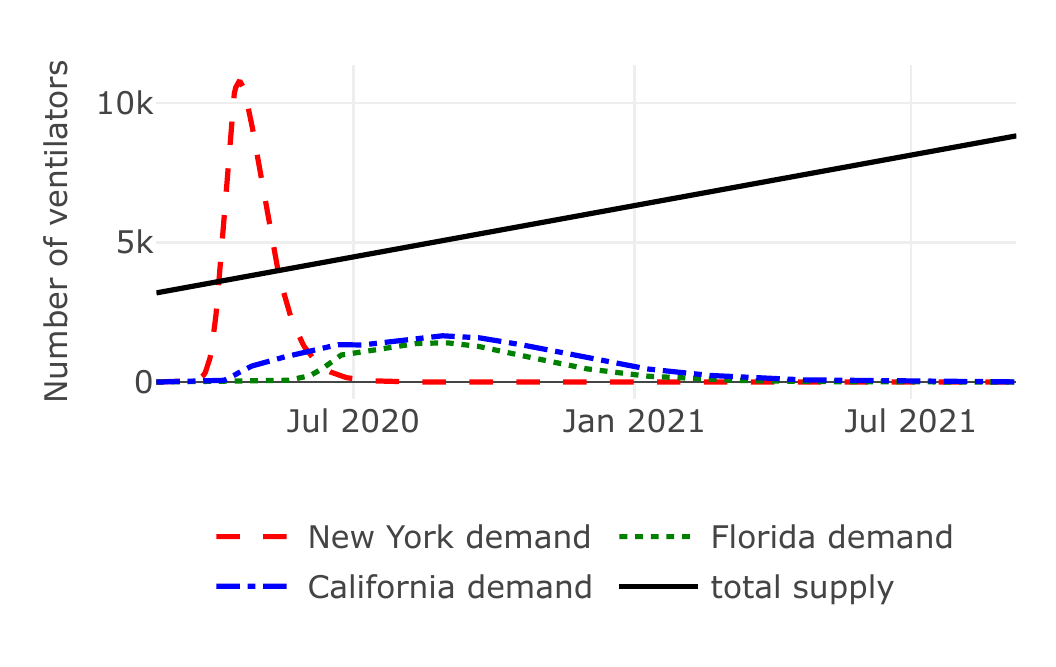}
\caption{Regional ventilator demands}
\label{fig:vent_allo_demands}
\end{subfigure}%
\begin{subfigure}{.5\textwidth}
\centering
\includegraphics[width=\linewidth]{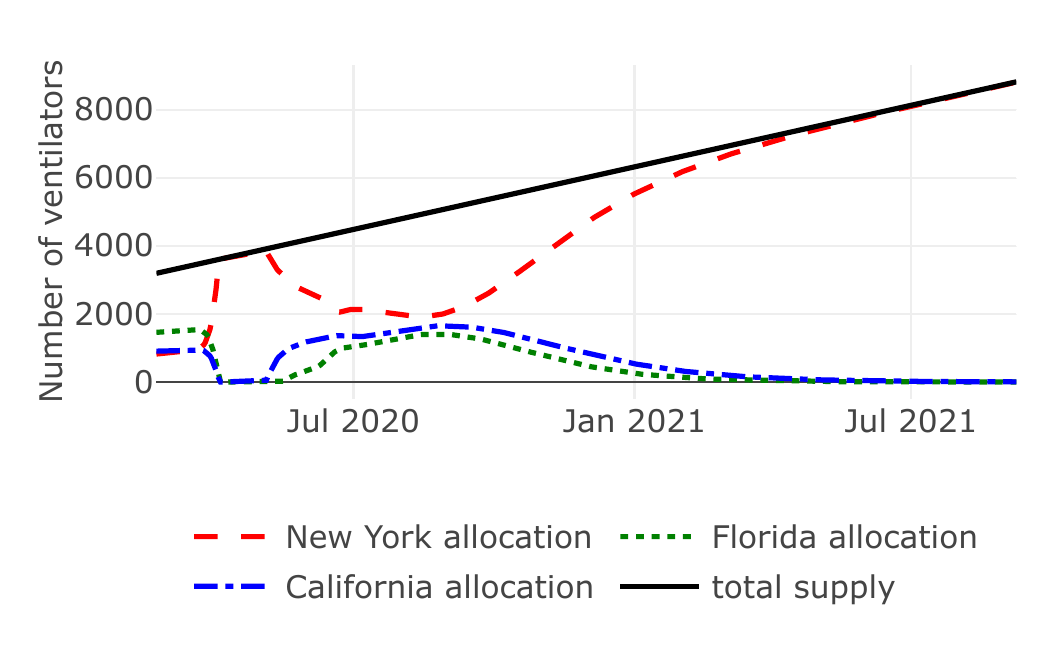}
\caption{Optimal ventilator allocations}
\label{fig:vent_allo_supplies}
\end{subfigure}
\caption{Optimal ventilator allocations in New York, Florida, and California}
\label{fig:vent_allo}
\end{figure}

\begin{figure}[h!]
\centering
\begin{subfigure}{.5\textwidth}
\centering
\includegraphics[width=\linewidth]{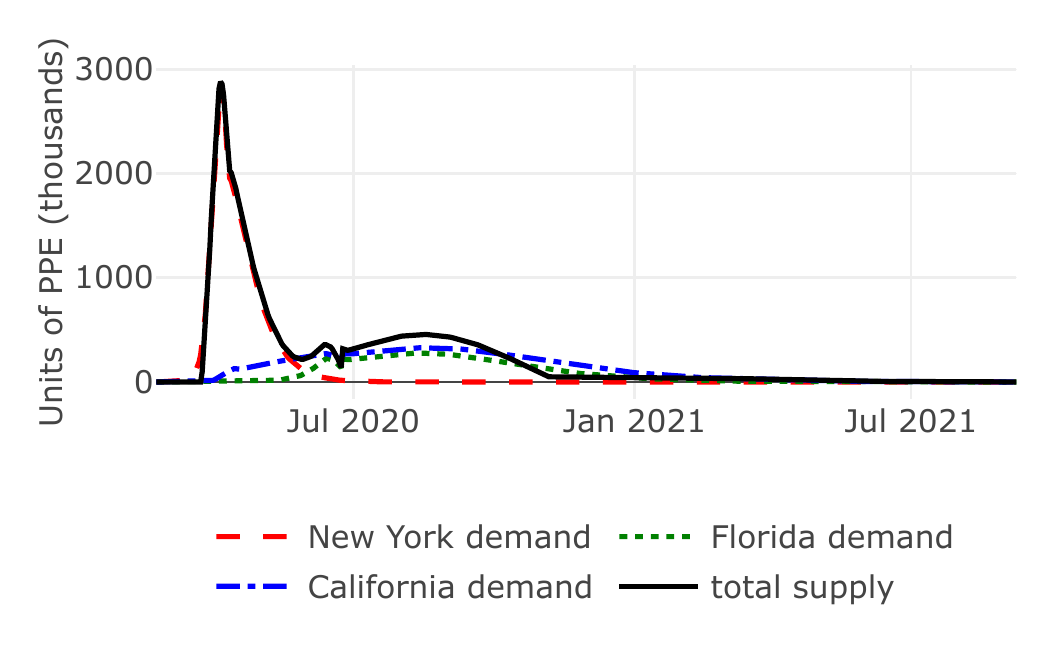}
\caption{Regional PPE demands}
\label{fig:ppe_allo_demands}
\end{subfigure}%
\begin{subfigure}{.5\textwidth}
\centering
\includegraphics[width=\linewidth]{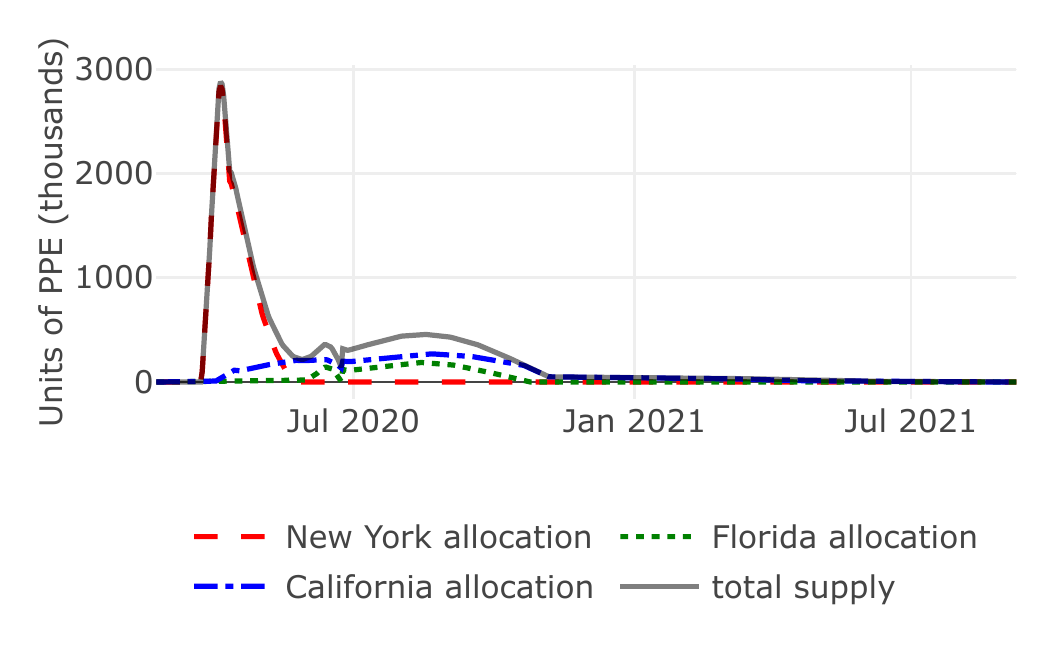}
\caption{Optimal PPE allocations}
\label{fig:ppe_allo_supplies}
\end{subfigure}
\caption{Optimal personal protective equipment allocations in New York, Florida, and California}
\label{fig:ppe_allo}
\end{figure}

\subsection*{Durable Resource: Ventilator}

Throughout this section we consider the allocation of existing resources in a healthcare system with $n$ regions during a pandemic that lasts for $m$ days. We always use the superscript $(i)$ to indicate quantities for the $i$-th region.
Bear in mind that there could still be aggregate shortage of supply for ventilators and PPE sets to all regions in the alliance. The central authority would have to take a holistic view of competing interests of participating regions. On each day in the pandemic, when the aggregate demand exceeds the aggregate supply, the central authority should choose to allocate resources taking into account spatial differences in demand and supply. This motivates the optimization model for ventilator allocations:
\begin{align*}
\underset{j=1,2,\dots,m}{\min_{K_j^{(i)}\geq 0;i=1,2,\dots,n;}}\;\; &\sum_{j=1}^{m}\sum_{i=1}^{n}\omega_j^{(i)}\left(\frac{\theta_j^{+(i)}}{2}\left(X_j^{\text{VEN}(i)}-K_j^{(i)}\right)_+^2+\frac{\theta_j^{-(i)}}{2}\left(X_j^{\text{VEN}(i)}-K_j^{(i)}\right)_-^2\right)\\\text{such that }&\sum_{i=1}^{n}K_j^{(i)}=K_j,\quad\text{for }j=1,2,\dots,m,
\end{align*}
where $\omega_j^{\left(i\right)}$ is a weight to the $j$-th day of the pandemic in the $i$-th allied region, $\theta_j^{+(i)}$ is an economic cost per squared unit of shortage, $\theta_j^{-(i)}$ is an opportunity cost per squared unit of oversupply. The quadratic terms $\theta_j^{\pm (i)}/2\left(X_j^{(i)}-K_j^{(i)}\right)^2$ represents the economic cost due to the demand-supply imbalance. Note that $\theta_j^{\pm(i)}/2$ measures the rate of increase in cost per unit, and hence $\theta_j^{\pm(i)}/2\left(X_j^{(i)}-K_j^{(i)}\right)$ represents the linear variable cost per unit. The variable cost in principle reflects the law of demand that the price increases with the quantity demanded. Therefore, the total cost is the product of variable cost per unit $\theta_j^{\pm(i)}/2\left(X_j^{(i)}-K_j^{(i)}\right)$ and the total unit of imbalance $\left(X_j^{(i)}-K_j^{(i)}\right)$. The economic cost is used to account for both potential loss of lives due to the lack of resources and the opportunity cost of idle medical sources due to oversupply. The structure of economic cost is used not only for its mathematical tractability, but also to penalize large imbalance of demand and supply. The weight $\omega_j^{(i)}$ can be used to measure the relative importance of resource allocation for region $i$ at time $t_j$ to other regions and time points. There are some examples of its applications under various criteria for resource pooling or rationing. For example, in a national contingency planning, where $X^{(i)}$ is used as predicted demand from each region, a metropolitan area with a large population may carry more weight than a rural area with a small region for political reasons. When hospitals have to ration limited resources, they may implement the strategy to maximize the benefit from treatment. In such a setup, $X^{(i)}$ represents the demand from a particular cohort. A decision maker may give higher weight to age cohorts with more remaining life years than age cohorts with less remaining years. It is also a common strategy to give priorities for access to medical resources to healthcare workers. In both cases, the set of weights $\omega^{(i)} _j$ reflects the management's priorities and preferences over time.


The constraint $\sum^n_{i=1} K^{(i)} _j=K_j$ indicates that resources allocated to different regions must add up to the total amount of supply available to the central authority. The evolution of supply $\{K_j, j=1,2,\dots, m\}$ is based on the centralized stockpiling strategy discussed in previous sections. The evolution of demand $\{X^{(i)} _j, i=1,\dots, n, j=1,2,\cdots,m\}$ can be based on forecasts from epidemiological models fitted to most recent local data.


\subsection*{Single-Use Resources: Personal Protective Equipment}

The allocation of single-use resources is similar to that of durable resources. The key difference lies in the amount of stockpile to be released each period. For durable resources, the central authority distributes the accumulated stock $K_j$ at any given time $j$. Because single-use resources cannot be reused, the central authority can only distribute incremental amount according to some distribution schedule. With this difference in mind, we formulate the allocation of single-use resources by an optimization problem.
\begin{align*}
\underset{j=1,2,\dots,m}{\min_{k_j^{(i)}\geq 0;i=1,2,\dots,n;}}\;\; &\sum_{j=1}^{m}\sum_{i=1}^{n}\omega_j^{(i)}\left(\frac{\theta_j^{+(i)}}{2}\left(X_j^{\text{PPE}(i)}-k_j^{(i)}\right)_+^2+\frac{\theta_j^{-(i)}}{2}\left(X_j^{\text{PPE}(i)}-k_j^{(i)}\right)_-^2\right)\\\text{such that }&\sum_{i=1}^{n}k_j^{(i)}=k_j,\quad\text{for }j=1,2,\dots,m.
\end{align*}
Note that the distribution amount $k_j$ could be determined prior to a pandemic by some contingency planning or during a pandemic by an adjusted distribution schedule. The ``Single-Use Resources: Personal Protective Equipment" section offers an example of how such a distribution schedule can be determined to take into temporal competition of single-use resources.


\subsection*{Holistic Allocation Algorithms}

Because both allocation problems take the same form, their analytical solutions can be derived in the same way. Because the allocation is done from period to period in the solution, we shall suppress the subscript $j$ for brevity. To simplify notation in the solution, we will use $X^{(i)}$ without the indicator of resource type for the demand in region $i$, and use $K^{(i)}$ for the quantity of allocated resources in region $i$. 

Here we discuss the analytical solutions to the optimization problems presented above, from which we can glean economic insights. The proofs of these solutions are relegated to Appendix \ref{app:alloc}. 
The central authority has to first determine whether or not there is a system wide surplus or shortage. The allocation strategy differs under these scenarios.

\subsubsection*{System wide surplus}

If there is an overall surplus in the healthcare system at time $j$, i.e., $K>\sum\limits_{r=1}^nX^{(r)}=X$, then only the economic cost for oversupply $\theta^{-(i)} $ applies and the optimal allocation of existing supply to the $i$-th region is given by
\begin{align}\label{fulldur}
    K^{(i)}=\left(1-\frac{\frac{1}{\omega^{(i)}\theta^{-(i)}}}{\sum\limits_{r=1}^{n}\frac{1}{\omega^{(r)}\theta^{-(r)}}}\right)X^{(i)}+\frac{\frac{1}{\omega^{(i)}\theta^{-(i)}}}{\sum\limits_{r=1}^{n}\frac{1}{\omega^{(r)}\theta^{-(r)}}}\left(K-\sum\limits_{r\neq i}X^{(r)}\right), \forall i=1,2,\cdots,n.
\end{align} Observe that the allocation formula \eqref{fulldur} has an explicit economic interpretation, which shows that the optimal supply for region $i$ results from a balance of two competing optimal solutions.
\begin{itemize}
\item {\bf Self-concerned optimal supply:} $X^{(i)}$

If region $i$ can ask for as much as it needs, then this amount shows the ideal supply in the best interest of the region alone. The demand and supply for all other regions are ignored in its consideration.

\item {\bf Altruistic optimal supply:} $K-\sum_{r=1;r\neq i}^{n}X^{(r)}$

If the region $i$ places interests of all other regions above its own, then the medical supply goes to other regions and region $i$ ends up with the leftover amount.
\end{itemize}
The central authority has the responsibility to mediate among regions competing for resources. The formula \eqref{fulldur} indicates that the optimality for region $i$ in consideration of the entire system is the weighted average of two extremes, namely the self-concerned optimal and the altruistic optimal supplies. It should be pointed that the average of two optimal supplies is determined by the harmonic weighting $\frac{1}{\omega^{(i)}\theta^{-(i)}}\Big/\sum\limits_{r=1}^{n}\frac{1}{\omega^{(r)}\theta^{-(r)}}$ as opposed to arithmetic weight $\omega^{(i)}\theta^{-(i)}\Big/\sum\limits_{r=1}^{n} \omega^{(r)}\theta^{-(r)} $. It is known in \citet{Chong2020} that in multi-objective Pareto optimality the harmonic weighting is always used for balancing competing interests of participants in a group whereas the arithmetic weighting serves the purpose of balancing competing objectives of the same participant. The pandemic resource allocation problem is in essence a model of Pareto optimality with regards to competing interests of members in a group. 

An alternative interpretation of formula \eqref{fulldur} can be obtained from the equivalent formula
\begin{equation}
K^{(i)}=X^{(i)}+B^{(i)}\left(K-\sum_{r=1}^{n}X^{(r)}\right),\qquad B^{(i)}=\frac{\frac{1}{\omega^{(i)}\theta^{-(i)}}}{\sum_{r=1}^{n}\frac{1}{\omega^{(r)}\theta^{-(r)}}}.\label{fulldur2}
\end{equation}
It follows from \eqref{fulldur2} that the allocated resource is always presented as an adjustment to the actual demand.  When there is a surplus in the system supply after optimal supplies have been fully distributed to all regions, then additional resource can be made available for region $i$ and each region obtains a portion determined by harmonic weighting. Observe that $\sum^n_{i=1} K^{(i)}=K$ as expected since $\sum^n_{i=1} B^{(i)}=1.$

\subsubsection*{System wide shortage}

If there is an overall shortage in the healthcare system at time $j$,  \textit{i.e.}, $K\le\sum_{r=1}^{n}X^{(r)}=X$, it turns out that the optimal allocation strategy is to deliver the resources where they are needed the most. We can summarize the algorithm in three steps:

\begin{itemize}
    \item[] {\bf Step 1: Demand ranking. } The first order of business is to sort regional demands $\{X^{(i)}, i=1, \dots, n\}$ in a descending order. We use the subscript $[i]$ to indicate the $i$-th largest order statistic, \textit{i.e.}, $X^{[1]}\ge \cdots\ge X^{[n]}\ge 0$. The ranking of regional demands determines the order in which the regions are considered for resource allocation in the next step.
    \item[] {\bf Step 2: Frugality test. }
The algorithm first tests cases that perform allocation rules in a similar way to \eqref{fulldur2}. For any fixed $I=1, \cdots,n$, consider the holistic allocation rule that provides for $I$ regions with largest demand by
\be \tilde{K}^{[i]} =X^{[i]}+\tilde{B}^{[i]}\left(K-\sum_{r=1}^{I}X^{[r]}\right),\qquad \tilde{B}^{[i]}=\frac{\frac{1}{\omega^{[i]}\theta^{+[i]}}}{\sum_{r=1}^{I}\frac{1}{\omega^{[r]}\theta^{+[r]}}}. \label{frgtest}\ee To find the optimal number $I$ of regions to provide support to, the algorithm ensures that the allocation rule should be frugal to meet the following criteria:
    \begin{itemize}
        \item[(i)] The total supply $K$ is only almost enough to meet the demands for all $I$ regions;
        \[K \le \sum_{r=1}^I X^{[r]}.\]
        \item[(ii)] When the allocation rule \eqref{frgtest} is forcefully applied to all regions, the  $I$ regions with highest demands should receive non-negative allocation and the rest of the group receive negative allocation.
        \[\tilde{K}^{[1]}, \cdots, \tilde{K}^{[I]} \ge 0>  \tilde{K}^{[I+1]}, \cdots, \tilde{K}^{[n]}.\]
    \end{itemize}
    There is a unique value of $I$ that passes the frugality test. As the aim of the strategy is to cover as many regions of highest demand as possible, the search algorithm stops after the total demand of $I$ regions exceeds the available supply. The algorithm would reach a rule that can be rewarding for those $I$ regions but discourages allocations to the rest.
\item[] {\bf Step 3: Holistic allocation.} Once the algorithm settles on the value of $I$, all existing resources are divided among the $I$ states according to the holistic allocation principle. In other words, the allocation of supply to the $i$-th region is given by 
\begin{align*}
    K^{[i]}&=\tilde{K}^{[i]},\qquad  \forall i=1,2,\cdots,I;\\
    K^{[I+1]}&=\cdots=K^{[n]}=0.
\end{align*} The general idea of the holistic allocation is illustrated in Figure \ref{fig:hos}. The burden of shortage $Y:=\sum^{I}_{r=1} X^{[r]} -K$ is carried by all $I$ states in proportion to their respective harmonic weight $\tilde{B}^{[i]}$. Therefore, each region receives its demand less its ``fair" portion of system-wide shortage, i.e. $X^{[i]}-\tilde{B}^{[i]}Y.$
\end{itemize}

\begin{figure}[htb]
\begin{center}
\scalebox{0.9}{
\begin{tikzpicture}
\node[text width=2cm] at (-6.7,1.25) {System wide total supply};
\draw (-8,-0.5) rectangle (1,0.5) node[pos=.5] {$K$};
\draw[dashed,pattern=north west lines] (1,-0.5) rectangle (1.7,0.5);
\draw[dashed] (1.7,-0.5) rectangle (5.8,0.5);
\draw[dashed,pattern=vertical lines] (5.8,-0.5) rectangle (6.7,0.5);
\draw[dashed, line width=2pt] (1,-0.5) rectangle (6.7,0.5) node[pos=.5] {Shortage};
\draw (1.35,1.5) circle (0.4cm and 0.4cm) node {$\tilde{B}^{[1]}$};
\draw[->] (1.35,0.5) -- (1.35,1);
\draw (3.7,1.5) node {$\cdots$};
\draw (6.2,1.5) circle (0.4cm and 0.4cm) node {$\tilde{B}^{[I]}$};
\draw[->] (6.2,0.5) -- (6.2,1);

\node[text width=2cm] at (-6.7,-1.25) {Regional demand};
\draw (-8,-3) rectangle (-5,-2) node[pos=.5] {$X^{[1]}$};
\draw (-5,-3) rectangle (4.5,-2) node[pos=.5] {$\dots$};
\draw (4.5,-3) rectangle (6.7,-2) node[pos=.5] {$X^{[I]}$};

\node[text width=2cm] at (-6.7,-3.75) {Allocation};
\draw (-8,-5.5) rectangle (-5.7,-4.5);
\draw[dashed,pattern=north west lines] (-5.7,-5.5) rectangle (-5,-4.5);
\draw (-5.7,-5.5) rectangle (-0.3,-4.5) node[pos=.5] {$\dots$};
\draw (-0.3,-5.5) rectangle (1,-4.5);
\draw[line width=2pt] (-8,-5.5) rectangle (1,-4.5);
\draw[dashed,pattern=vertical lines] (1,-5.5) rectangle (1.9,-4.5);

\draw [decorate,decoration={brace,mirror,amplitude=10pt},xshift=-4pt,yshift=0pt]
(-7.9,-6) -- (-5.6,-6) node [black,midway,yshift=-0.6cm] 
{\footnotesize $K^{[1]} $};
\draw (-3,-6.6) node {$\cdots$};
\draw [decorate,decoration={brace,mirror,amplitude=10pt},xshift=-4pt,yshift=0pt]
(-0.2,-6) -- (1.1,-6) node [black,midway,yshift=-0.6cm] 
{\footnotesize $K^{[I]} $};
\draw [decorate,decoration={brace,mirror,amplitude=10pt},xshift=-4pt,yshift=0pt]
(-7.9,-7) -- (1.1,-7) node [black,midway,yshift=-0.6cm] 
{\footnotesize $K$};

\draw[densely dotted] (-8,-3) -- (-8,-4.5);
\draw[densely dotted] (-5,-3) -- (-5,-4.5);
\draw[dashed, semithick] (-5,-3) -- (-5.7,-4.5);
\draw[densely dotted] (4.5,-3) -- (-0.3,-4.5);
\draw[densely dotted] (6.7,-3) -- (1.9,-4.5);
\draw[densely dotted] (-8,-0.5) -- (-8,-4.5);
\draw[densely dotted] (1,-0.5) -- (1,-4.5);
\draw[dashed, semithick] (6.7,-3) -- (1,-4.5);
\end{tikzpicture}}
\end{center}
\caption{Holistic allocation of resources in face of system wide shortage}
\label{fig:hos}
\end{figure}
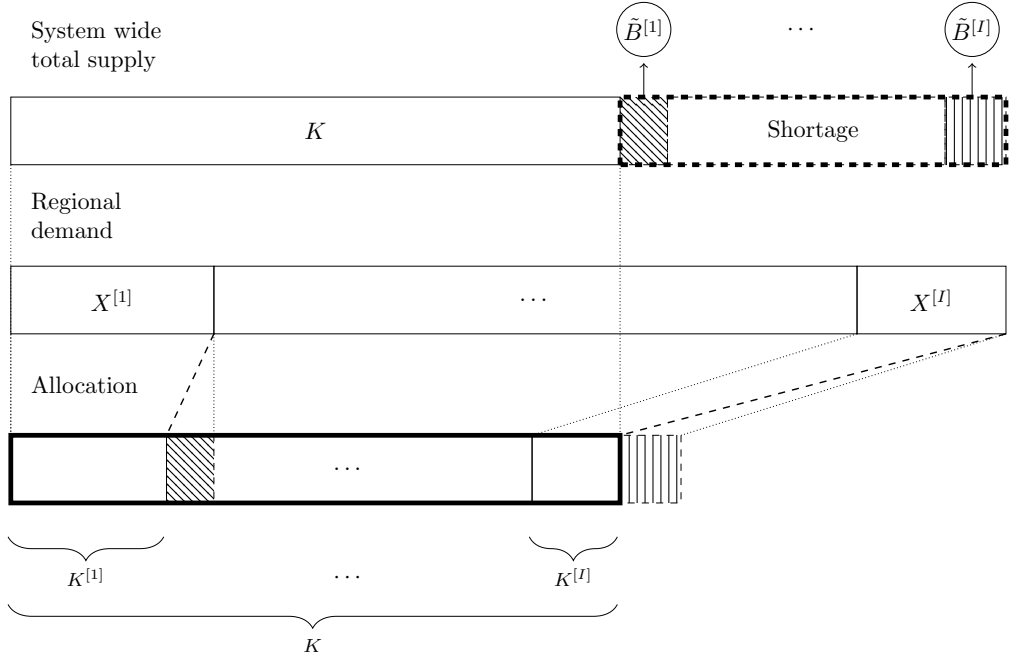

The solution to the three-state alliance for pooling ventilators and PPE sets  are shown in Figures \ref{fig:vent_allo} and \ref{fig:ppe_allo}. Figure \ref{fig:vent_allo_supplies} depicts the case of optimal allocation strategies for ventilators, which confirms the intuition in the ``Case Study" Section. In April, the central authority could have optimally reallocated all available aggregate ventilators in the alliance to NY; This is owing to the fact that NY has the highest demand of all three states. By May and June, ventilators in NY could have gradually reallocated to both FL and CA; in July and August, with the reallocated resources from NY, both FL and CA should have experienced no shortage of ventilators at all. Figure \ref{fig:ppe_allo_supplies} illustrates optimal allocation strategies for PPE sets in the case study. In April, the central authority could distribute stockpiled PPE sets and have all sent to NY for emergency response; As the pandemic dies down for NY and picks up for CA and FL, the resources are more evenly spread in June. By August, all PPE sets should be released to CA and FL both with high demands.

\section*{Conclusions and Limitations} \label{sec:conclusion}
The COVID-19 pandemic has placed extraordinary demands and constraints on public healthcare systems, exposing many problems such as the lack of adequate planning and coordination among others. This paper investigates what could have been done better to reduce the imbalance of medical resources demand and supply. Inspired by classical theory of risk aggregation and capital allocation, this paper proposes a three-pillar resources planning and allocation framework --- demand forecast, centralized stockpiling and distribution, and centralized resources allocation. This paper further develops a novel spatio-temporal balancing of resources and can potentially used by public policymakers as quantitative basis for making informed decisions on planning, funding and rationing of critical resources. 

It should be pointed out that the paper focuses on managerial insights that can be drawn from the optimization framework and its analytical solutions. There are admittedly a number of limitations in the hypothetical example of a three-state resources pooling arrangement. The three states are chosen for the most drastic effects of planning and allocation of existing resources for the purpose of illustration. It is beyond the scope of this paper to consider the political reality that may prevent such arrangements. In theory, the methodology can be applied to the actual voluntary coalition formed by six northwestern states in the US, although the coalition was formed largely to avoid price competition in government procurement. Another limitation of this example is the egalitarian approach to shortages in different regions, which ignores ethical issues that may arise from freely moving resources from one region to another. While it may be economically optimal to deliver all resources in the system to the place where they are needed the most, it may be politically challenging to leave other places with less severe shortage without support. A potential remedy would be to introduce additional constraints in the optimization problems that require some minimal support for each region.



\newpage

\printbibliography

\appendix

\section{Analytical Solutions and Proofs}
\subsection{Stockpiling of durable resources}\label{appendix_b1}
The optimization problem for the stockpiling of durable resources is as follows:
\begin{align*}
    \min_{K_0\ge 0}\sum_{j=1}^{m}\omega_j\left(\frac{\theta_j^+}{2}(X_j-(K_0+aj))_+^2+\frac{\theta_j^-}{2}(X_j-(K_0+aj))_-^2+c_j(K_0+aj)\right)+c_0K_0.
\end{align*}

\begin{thm}
    Let $Y_j=X_j-aj, \forall j=1,2,\cdots,m$. Let $S=\sum_{j=1}^m\omega_jc_{j}+c_0.$ Let $Y_{[1]}\le Y_{[2]}\le \cdots\le Y_{[m]}$ be the increasingly ordered sequence of $Y_1,Y_2,\cdots,Y_m$. Let $J=1,2,\cdots,m$ such that 
    \[\sum_{j=1}^{J-1}\omega_{[j]}{\theta_{[j]}^-}(Y_{[j]}-Y_{[J]})+\sum_{j=J}^{m}\omega_{[j]}{\theta_{[j]}^+}(Y_{[j]}-Y_{[J]})\le S\le \sum_{j=1}^{J-1}\omega_{[j]}{\theta_{[j]}^-}(Y_{[j]}-Y_{[J-1]})+\sum_{j=J}^{m}\omega_{[j]}{\theta_{[j]}^+}(Y_{[j]}-Y_{[J-1]}),\]
    where we define that when $J=1$, 
    \[\sum_{j=1}^{J-1}\omega_{[j]}{\theta_{[j]}^-}(Y_{[j]}-Y_{[J]})+\sum_{j=J}^{m}\omega_{[j]}{\theta_{[j]}^+}(Y_{[j]}-Y_{[J]})=\sum_{j=1}^{m}\omega_{[j]}{\theta_{[j]}^+}(Y_{[j]}-Y_{[1]}),\]
    \[\sum_{j=1}^{J-1}\omega_{[j]}{\theta_{[j]}^-}(Y_{[j]}-Y_{[J-1]})+\sum_{j=J}^{m}\omega_{[j]}{\theta_{[j]}^+}(Y_{[j]}-Y_{[J-1]}) = \infty.\]
    Let \[K_0'=\frac{\sum_{j=1}^{J-1}\omega_{(j)}{\theta_{(j)}^-}Y_{(j)}+\sum_{j=J}^{m}\omega_{(j)}{\theta_{(j)}^+}Y_{(j)}-S}{\sum_{j=1}^{J-1}\omega_{(j)}{\theta_{(j)}^-}+\sum_{j=J}^{m}\omega_{(j)}{\theta_{(j)}^+}}.\]
    If $K_0'<0$, then the optimal initial stockpile, which minimizes the objective function above, $K_0^*=0$, and if $K_0'\ge 0$, then $K_0^*=K_0'$.
\end{thm}

\begin{proof}
Let $Y_j=X_j-aj, \forall j=1,2,\cdots,m$.
Let $Y_{[1]}\le Y_{[2]}\le \cdots\le Y_{[m]}$ be the increasingly ordered sequence of $Y_1,Y_2,\cdots,Y_m$. In that case, $Y_j$ represents the daily shortage when the initial stockpile is completely missing. Then the objective function becomes
\begin{equation*}F(K_0):=\sum_{j=1}^{m}\omega_{[j]}\left(\frac{\theta_{[j]}^+}{2}(Y_{[j]}-K_0)_+^2+\frac{\theta_{[j]}^-}{2}(K_0-Y_{[j]})_+^2\right)+\left(\sum_{j=1}^m\omega_jc_{j}+c_0\right)K_0+\sum_{j=1}^m\omega_jc_jaj.\\
\end{equation*}
Note that $F(K_0)$ is a convex function in $K_0$ for any $K_0\in \R$.
Let 
\begin{align*}
G(K_0)&=\sum_{j=1}^{m}\omega_{[j]}\left(\frac{\theta_{[j]}^+}{2}(Y_{[j]}-K_0)_+^2+\frac{\theta_{[j]}^-}{2}(K_0-Y_{[j]})_+^2\right),\\
S&=\sum_{j=1}^m\omega_jc_{j}+c_0,
\end{align*}
where $G(\cdot)$ is convex in $K_0$.

The first order derivatives of $G$ at $Y_{[1]}$ and $Y_{[m]}$ are as follows:

\[G'(Y_{[1]})=\sum_{j=1}^{m}\omega_{[j]}{\theta_{[j]}^-}(Y_{[1]}-Y_{[j]})\le 0.\]
\[G'(Y_{[m]})=\sum_{j=1}^{m}\omega_{[j]}{\theta_{[j]}^+}(Y_{[j]}-Y_{[m]})(-1)=\sum_{j=1}^{m}\omega_{[j]}{\theta_{[j]}^+}(Y_{[m]}-Y_{[j]})\ge 0.\]
Then $\exists$ $\tilde{K_0}\in\left[Y_{[1]},Y_{[m]}\right]$ such that $G'(\tilde{K_0})=0$. Since $S = \sum_{j=1}^{m}\omega_jc_j+c_0\ge 0$, $\exists$ $K_0'\in(-\infty,Y_{[m]}]$ such that $F'(K_0')=0$, \textit{i.e.}, $\exists$ $J=1,\cdots,m$ such that $K_0'\in\left[Y_{[J-1]},Y_{[J]}\right]$, where we define $Y_{(0)}=-\infty$, and $F'(K_0')=0$.

Here, we are relaxing the constraint that the initial stockpile has to be non-negative, but we will add it back in the end. The key observation is that given all the $(Y_{[j]})_{j \in \{0,\dots, m\}}$, we can always find a $K_0'$ between two adjacent $Y_{[j]}$'s, $Y_{[J-1]}$ and $Y_{[J]}$, that minimize the objective function.
That is, \[F'(K_0')=\sum_{j=1}^{J-1}\omega_{[j]}{\theta_{[j]}^-}(K_0'-Y_{[J]})+\sum_{j=J}^{m}\omega_{[j]}{\theta_{[j]}^+}(Y_{[j]}-K_0')(-1)+S=0.\]
Then \[K_0'=\frac{\sum_{j=1}^{J-1}\omega_{[j]}{\theta_{[j]}^-}Y_{[j]}+\sum_{j=J}^{m}\omega_{[j]}{\theta_{[j]}^+}Y_{[j]}-S}{\sum_{j=1}^{J-1}\omega_{[j]}{\theta_{[j]}^-}+\sum_{j=J}^{m}\omega_{[j]}{\theta_{[j]}^+}}.\]

Therefore, the condition that $K_0'\in[Y_{[J-1]},Y_{[J]}]$ is given by 
\[Y_{[J-1]}\le \frac{\sum_{j=1}^{J-1}\omega_{[j]}{\theta_{[j]}^-}Y_{[j]}+\sum_{j=J}^{m}\omega_{[j]}{\theta_{[j]}^+}Y_{[j]}-S}{\sum_{j=1}^{J-1}\omega_{[j]}{\theta_{[j]}^-}+\sum_{j=J}^{m}\omega_{[j]}{\theta_{[j]}^+}} \le Y_{[J]}.\]
Or equivalently, $\exists$ $J=1,2,\cdots,m$, such that 
\[\sum_{j=1}^{J-1}\omega_{[j]}{\theta_{[j]}^-}(Y_{[j]}-Y_{[J]})+\sum_{j=J}^{m}\omega_{[j]}{\theta_{[j]}^+}(Y_{[j]}-Y_{[J]})\le S\le \sum_{j=1}^{J-1}\omega_{[j]}{\theta_{[j]}^-}(Y_{[j]}-Y_{[J-1]})+\sum_{j=J}^{m}\omega_{[j]}{\theta_{[j]}^+}(Y_{[j]}-Y_{[J-1]}),\] 
wherein if $J=1$, we define
    \[\sum_{j=1}^{J-1}\omega_{[j]}{\theta_{[j]}^-}(Y_{[j]}-Y_{[J]})+\sum_{j=J}^{m}\omega_{[j]}{\theta_{[j]}^+}(Y_{[j]}-Y_{[J]})=\sum_{j=1}^{m}\omega_{[j]}{\theta_{[j]}^+}(Y_{[j]}-Y_{[1]}),\]
    \[\sum_{j=1}^{J-1}\omega_{[j]}{\theta_{[j]}^-}(Y_{[j]}-Y_{[J-1]})+\sum_{j=J}^{m}\omega_{[j]}{\theta_{[j]}^+}(Y_{[j]}-Y_{[J-1]}) = \infty.\]
    
Finally, given the non-negativity of stockpile, if $K_0'<0$, then the optimal initial stockpile $K_0^*=0$, and if $K_0'\ge 0$, then $K_0^*=K_0'$.
\end{proof}

\subsection{Centralized resources allocation} \label{app:alloc}
The optimization problem for allocating resources amount regions is as follows:
\begin{align*}
\underset{j=1,2,\dots,m}{\min_{K_j^{(i)}\geq 0;i=1,2,\dots,n;}}\;\; &\sum_{j=1}^{m}\sum_{i=1}^{n}\omega_j^{(i)}\left(\frac{\theta_j^{+(i)}}{2}\left(X_j^{(i)}-K_j^{(i)}\right)_+^2+\frac{\theta_j^{-(i)}}{2}\left(X_j^{(i)}-K_j^{(i)}\right)_-^2\right)\\\text{such that }&\sum_{i=1}^{n}K_j^{(i)}=K_j,\quad\text{for }j=1,2,\dots,m,
\end{align*}

\begin{thm} The optimization above is done from period to period, and thus to simplify notation, the time indicator $j$ can be dropped at each time point. let $X^{[1]}\ge \dots\ge X^{[n]} > 0$ be the decreasingly ordered sequence of $X^{(1)}, X^{(2)}, \dots, X^{(n)}$. 

If $K>\sum\limits_{r=1}^nX^{(r)}=X$, then
\begin{align*}
    K^{(i)}=\left(1-\frac{\frac{1}{\omega^{(i)}\theta^{-(i)}}}{\sum\limits_{r=1}^{n}\frac{1}{\omega^{(r)}\theta^{-(r)}}}\right)X^{(i)}+\frac{\frac{1}{\omega^{(i)}\theta^{-(i)}}}{\sum\limits_{r=1}^{n}\frac{1}{\omega^{(r)}\theta^{-(r)}}}\left(K-\sum\limits_{r\neq i}X^{(r)}\right), \forall i=1,2,\dots,n.
\end{align*}

If $K\le\sum_{r=1}^{n}X^{(r)}=X$, we can find an $I=1,2,\dots,n$ such that
$K\le \sum_{r=1}^{I}X^{[r]}$,
\begin{align*}
X^{[i]}&\ge \frac{\frac{1}{\omega^{[i]}\theta^{+[i]}}}{\sum\limits_{r=1}^{I}\frac{1}{\omega^{[r]}\theta^{+[r]}}}\left(\sum\limits_{r=1}^IX^{[r]}-K\right),\forall i=1,\dots,I,\\
X^{[i]}&<\frac{\frac{1}{\omega^{[i]}\theta^{+[i]}}}{\sum\limits_{r=1}^{I}\frac{1}{\omega^{[r]}\theta^{+[r]}}}\left(\sum\limits_{r=1}^IX^{[r]}-K\right), \forall i=I+1,I+2,\dots,n.
\end{align*}
Then $K^{[I+1]}=\dots=K^{[n]}=0$. 
\begin{align*}
    K^{[i]}=\left(1-\frac{\frac{1}{\omega^{[i]}\theta^{+[i]}}}{\sum\limits_{r=1}^{I}\frac{1}{\omega^{[r]}\theta^{+[r]}}}\right)X^{[i]}+\frac{\frac{1}{\omega^{[i]}\theta^{+[i]}}}{\sum\limits_{r=1}^{I}\frac{1}{\omega^{[r]}\theta^{+[r]}}}\left(K-\sum\limits_{r=1,r\neq i}^IX^{[r]}\right), \forall i=1,\dots,I.
\end{align*}
\end{thm}

\begin{proof}
At each time point $j$, we want to solve the following optimization problem,
\begin{align*}
    \min_{K^{(i)};i=1,2,\cdots,n}&\sum_{i=1}^{n}\omega^{(i)}\left(\frac{\theta^{+(i)}}{2}(X^{(i)}-K^{(i)})^2_++\frac{\theta^{-(i)}}{2}(K^{(i)}-X^{(i)})^2_+\right)\\
    \text{such that }\ &\sum_{i=1}^{n}K^{(i)}=K, K^{(i)}\ge 0, \forall i=1,2,\cdots,n.
\end{align*}

\begin{casesp}
\item \textbf{Coexisting surpluses and shortages} 

First, let us consider the case in which some regions are having surpluses, whereas other regions are experiencing shortages at the same time. We shall show below that this case is impossible regardless of the non-negative constraints.
That is, suppose that $K^{(1)},K^{(2)},\dots,K^{(n)}$ lie locally in a feasible set,
such that $I$ of them satisfy $K^{(i)}>X^{(i)}$, where $I =1,2,\dots,n-1$. The remaining  $n-I$ of them satisfy $K^{(i)}\leq X^{(i)}$. Without loss of generality, assume the first $I$ of $K^{(i)}$ are in the former group. Then the local problem becomes
\begin{align*}
    \min_{K^{(i)};i=1,2,\cdots,n}&\sum_{i=I+1}^{n}\omega^{(i)}\frac{\theta^{+(i)}}{2}(X^{(i)}-K^{(i)})^2+\sum_{i=1}^{I}\omega^{(i)}\frac{\theta^{-(i)}}{2}(X^{(i)}-K^{(i)})^2\\ \text{such that }\ &\sum_{i=1}^{n}K^{(i)}=K. 
\end{align*}

The solution to this problem is given by \begin{align*}
    K^{(i)}=\left(1-\frac{\frac{1}{\omega^{(i)}\theta^{-(i)}}}{\sum\limits_{r=1}^{I}\frac{1}{\omega^{(r)}\theta^{-(r)}}+\sum\limits_{r=I+1}^n\frac{1}{\omega^{(r)}\theta^{+(r)}}}\right)X^{(i)}+\frac{\frac{1}{\omega^{(i)}\theta^{-(i)}}}{\sum\limits_{r=1}^{I}\frac{1}{\omega^{(r)}\theta^{-(r)}}+\sum\limits_{r=I+1}^n\frac{1}{\omega^{(r)}\theta^{+(r)}}}\left(K-\sum\limits_{r\neq i}X^{(r)}\right),
\end{align*}
for $i =1,2,\dots,I$, and 
\begin{align*}
    K^{(i)}=\left(1-\frac{\frac{1}{\omega^{(i)}\theta^{+(i)}}}{\sum\limits_{r=1}^{I}\frac{1}{\omega^{(r)}\theta^{-(r)}}+\sum\limits_{r=I+1}^n\frac{1}{\omega^{(r)}\theta^{+(r)}}}\right)X^{(i)}+\frac{\frac{1}{\omega^{(i)}\theta^{+(i)}}}{\sum\limits_{r=1}^{I}\frac{1}{\omega^{(r)}\theta^{-(r)}}+\sum\limits_{r=I+1}^n\frac{1}{\omega^{(r)}\theta^{+(r)}}}\left(K-\sum\limits_{r\neq i}X^{(r)}\right),
\end{align*}
for $i =I+1,\dots,n$. 

However, $\forall i = 1,2\dots,I$, $K^{(i)}>X^{(i)}$, and $\forall i = I+1,\dots,n$, $ K^{(i)}\le X^{(i)}$, which implies ${K>\sum_{r=1}^nX^{(r)}}$ and $K\le\sum_{r=1}^nX^{(r)}$, and thus we have a contradiction. This result shows that it is impossible for some regions to have surpluses while other regions are experiencing shortages. Therefore, it suffices to only consider the scenarios that there is a system wide surplus and that there is a system wide shortage.

\item \textbf{System wide surplus}

If there is a system wide surplus, \textit{i.e.}, all regions have surpluses, then the problem becomes 
\begin{align*}
    \min_{K^{(i)};i=1,2,\cdots,n}&\sum_{i=1}^{n}\omega^{(i)}\frac{\theta^{-(i)}}{2}(K^{(i)}-X^{(i)})^2\\ \text{such that }\ &\sum_{i=1}^{n}K^{(i)}=K, \\
    & K^{(i)}>X^{(i)}, \forall i=1,\dots,n.
\end{align*}
The solution is given by
\begin{align*}
    K^{(i)}=\left(1-\frac{\frac{1}{\omega^{(i)}\theta^{-(i)}}}{\sum\limits_{r=1}^{n}\frac{1}{\omega^{(r)}\theta^{-(r)}}}\right)X^{(i)}+\frac{\frac{1}{\omega^{(i)}\theta^{-(i)}}}{\sum\limits_{r=1}^{n}\frac{1}{\omega^{(r)}\theta^{-(r)}}}\left(K-\sum\limits_{r\neq i}X^{(r)}\right).
\end{align*}

For this result to hold, we only need the condition that $K>\sum_{r=1}^nX^{(r)}=X$. 
Due to uniqueness, under this condition, this $K^{(i)}$ is optimal.

\item \textbf{System wide shortage}

It remains to solve the case that there is a system wide shortage, \textit{i.e.}, all regions have shortages. In that case, the problem becomes 
\begin{align*}
    \min_{K^{(i)};i=1,2,\cdots,n}&\sum_{i=1}^{n}\omega^{(i)}\frac{\theta^{+(i)}}{2}(X^{(i)}-K^{(i)})^2\\ \text{such that }\ &\sum_{i=1}^{n}K^{(i)}=K; \\
    &0 \le K^{(i)}\le X^{(i)}, \forall i=1,\cdots,n.
\end{align*}
Apart from the given condition that $\sum_{i=1}^nK^{(i)}=K$, there are additional inequality constraints in this optimization problem. Hence, we make use of Karush-Kuhn–Tucker (KKT) conditions as follows.
\begin{align*}
    &\omega^{(i)}\theta^{+(i)}(K^{(i)}-X^{(i)})-\lambda_1^{(i)}+\lambda_2^{(i)}+\mu=0,\\
    &0\le K^{(i)}\le X^{(i)},\\
    &\lambda_1^{(i)}\ge 0, \quad \lambda_1^{(i)}K^{(i)}=0,\\
    &\lambda_2^{(i)}\ge 0, \quad \lambda_2^{(i)}(K^{(i)}-X^{(i)})=0,
\end{align*}
for all $i = 1,\cdots, n$. Regarding the values of $\lambda_1^{(i)}$ and $\lambda_2^{(i)}$, we can consider the following four mutually exclusive cases,
\begin{enumerate}
    \item[(i)] $\lambda_1^{(i)}>0$ for some $i = 1, \cdots, n$, and $\lambda_2^{(i)}>0$ for some $i = 1, \cdots, n$;
    \item[(ii)] $\lambda_1^{(i)}=0$ for all $i = 1, \cdots, n$, and $\lambda_2^{(i)}>0$ for some $i = 1, \cdots, n$;
    \item[(iii)] $\lambda_1^{(i)}>0$ for some $i = 1, \cdots, n$, and $\lambda_2^{(i)}=0$ for all $i = 1, \cdots, n$;
    \item[(iv)] $\lambda_1^{(i)}=0$ for all $i = 1, \cdots, n$, and $\lambda_2^{(i)}=0$ for all $i = 1, \cdots, n$.
\end{enumerate}
Each of them will be discussed as follows. 
\begin{casesp}
    \item $\lambda_1^{(i)}>0$ for some $i = 1, \cdots, n$, and $\lambda_2^{(i)}>0$ for some $i = 1, \cdots, n$.
    
    We will show that this case will lead to a contradiction, and thus is impossible. Because the ordering of $\lambda_1^{(i)}$ and $\lambda_2^{(i)}$ does not affect the conditions, for simplicity, we rearrange them in such a way that there is an $I = 1,2,\dots, n-1$, and an $\tilde{I} = 1,2,\dots, n-1$, for which 

\begin{align*}
    \lambda_1^{[i]}>0&, \quad \forall i=1,\dots, I, \\
    \lambda_1^{[i]}=0&, \quad \forall i=I+1,\dots, n, \\
    \lambda_2^{[i]}=0&, \quad \forall i=1,\dots, n-\tilde{I}, \\
    \lambda_2^{[i]}>0&, \quad \forall i = n-\tilde{I}+1,\dots, n,
\end{align*}
where $[i]$ are indices after the rearrangement.

We also have the condition that $I\le n-\tilde{I}$, because for each $i$, $\lambda_1^{[i]} > 0$ implies $K^{[i]} = 0$, which further implies $\lambda_2^{[i]}=0$. Therefore, the number of $\lambda_2^{[i]}$ that are equal to 0, \textit{i.e.}, $n-\tilde{I}$, is at least the number of $\lambda_1^{[i]}$ that are greater than 0, \textit{i.e.}, $I$. Then by the complementary slackness conditions in the KKT conditions above,  $K^{[i]}=0, \forall i=1,2,\cdots,I$ and,  $K^{[i]}=X^{[i]}, \forall i=n-\tilde{I}+1,\cdots,n$.

Then, the KKT conditions are simplified.
\begin{align*}
    -\omega^{[i]}\theta^{+[i]}X^{[i]}-\lambda_1^{[i]}+\mu=0&, \quad \forall i = 1,\dots, I\\
\omega^{[i]}\theta^{+[i]}(K^{[i]}-X^{[i]})+\mu=0&, \quad \forall i = I+1, \dots, n-\tilde{I}\\
\lambda_2^{[i]}+\mu=0&, \quad \forall i = n-\tilde{I}+1, \dots, n\\
0\le K^{[i]}\le X^{[i]}&,  \quad \forall i=I+1,\dots,n-\tilde{I} \\
\sum_{r=I+1}^{n-\tilde{I}}K^{[r]}=K-\sum_{r=n-\tilde{I}+1}^nX^{[r]}&,
\end{align*}
from which we can observe the following contradiction that,
\begin{align*}
    &\mu=-\lambda_2^{[n]}<0 \\
    &\mu=\omega^{[1]}\theta^{+[1]}X^{[1]}+\lambda_1^{[1]}>0.
\end{align*}

Therefore, we can conclude that if $\tilde{I}=1,\cdots,n-1$, then $I=0$ or $n$, and its contrapositive is also true, which states that if $I=1,\cdots,n-1$, then $\tilde{I}=0$ or $n$. These two statements correspond the second and third cases respectively, and they are considered as follows.

\item $\lambda_1^{(i)}=0$ for all $i = 1, \cdots, n$, and $\lambda_2^{(i)}>0$ for some $i = 1, \cdots, n$.

As in the previous case, We use the rearranged $\lambda_1^{[i]}$ and $\lambda_2^{[i]}$, so now $\lambda_2^{[i]} = 0$, for all $i = 1, \dots, n-\tilde{I}$, and $\lambda_2^{[i]} > 0$, for all $i = n-\tilde{I}+1,\dots,n$, for some $\tilde{I}=1,2,\dots,n-1$. This implies $K^{[i]}=X^{[i]}$, for all $i = n-\tilde{I}+1,\dots,n$.

In this case, the KKT conditions become
\begin{align*}
\omega^{[i]}\theta^{+[i]}(K^{[i]}-X^{[i]})+\mu=0&, \quad \forall i = 1, \dots, n-\tilde{I}\\
\lambda_2^{[i]}+\mu=0&, \quad \forall i = n-\tilde{I}+1, \dots, n\\
0\le K^{[i]}\le X^{[i]}&, \quad \forall i=1,\cdots,n-\tilde{I} \\
\sum_{r=1}^{n-\tilde{I}}K^{(r)}=K-\sum_{r=n-\tilde{I}+1}^nX^{(r)}&.
\end{align*}

By solving this system, we get
\begin{align*}
    K^{[i]}&=\left(1-\frac{\frac{1}{\omega^{[i]}\theta^{+[i]}}}{\sum\limits_{r=1}^{n-\tilde{I}}\frac{1}{\omega^{[r]}\theta^{+[r]}}}\right)X^{[i]}+\frac{\frac{1}{\omega^{[i]}\theta^{+[i]}}}{\sum\limits_{r=1}^{n-\tilde{I}}\frac{1}{\omega^{[r]}\theta^{+[r]}}}\left(K-\sum\limits_{r\neq i}^{n}X^{[r]}\right)\\
    &=X^{[i]}+\frac{\frac{1}{\omega^{[i]}\theta^{+[i]}}}{\sum\limits_{r=1}^{n-\tilde{I}}\frac{1}{\omega^{[r]}\theta^{+[r]}}}\left(K-\sum_{r=1}^{n}X^{[r]}\right), \quad \forall i=1,2,\cdots,n-\tilde{I}\\
    \mu&=\frac{1}{\sum\limits_{r=1}^{n-\tilde{I}}\frac{1}{\omega^{[r]}\theta^{+[r]}}}\left(\sum\limits_{r=1}^nX^{[r]}-K\right)\\
    \lambda_2^{[i]}&=-\mu=-\frac{1}{\sum\limits_{r=1}^{n-\tilde{I}}\frac{1}{\omega^{[r]}\theta^{+[r]}}}\left(\sum\limits_{r=1}^nX^{[r]}-K\right),\quad \forall i=n-\tilde{I}+1,\cdots,n.
\end{align*}

By the system wide shortage assumption, $\sum\limits_{r=1}^nX^{[r]}-K \geq 0$, and therefore, $\lambda_2^{[i]}\leq 0$ for all $i=n-\tilde{I}+1,\cdots,n$. But this contradicts the assumption that $\lambda_2^{[i]} > 0$ for all $i=n-\tilde{I}+1,\cdots,n$. Therefore, we can tell that this is another impossible case.

\item \label{case:3.3} $\lambda_1^{(i)}>0$ for some $i = 1, \cdots, n$, and $\lambda_2^{(i)}=0$ for all $i = 1, \cdots, n$.

Again, for this case, we use the rearranged $\lambda_1^{[i]}$ and $\lambda_2^{[i]}$ for $i = 1, \dots, n$. And we assume that there is an $I=1,2,\dots,n-1$, such that $\lambda_1^{[i]} > 0$ for $i = 1, \dots, I$, and $\lambda_1^{[i]} = 0$ for $i = I+1, \dots, n$. This implies $K^{[i]}=0$ for $i = 1, \dots, I$, and we get the following conditions, 
\begin{align*}
-\omega^{[i]}\theta^{+[i]}X^{[i]}-\lambda_1^{[i]}+\mu=0&, \quad \forall i = 1, \dots, I\\
\omega^{[i]}\theta^{+[i]}(K^{[i]}-X^{[i]})+\mu=0&, \quad \forall i = I+1, \dots, n\\
0\le K^{[i]}\le X^{[i]}&, \quad \forall i=I+1,\cdots,n \\
\sum_{i=I+1}^nX^{[r]}=K&.
\end{align*}
They together give us
\begin{align*}
    K^{[i]}&=\left(1-\frac{\frac{1}{\omega^{[i]}\theta^{+[i]}}}{\sum\limits_{r=I+1}^{n}\frac{1}{\omega^{[r]}\theta^{+[r]}}}\right)X^{[i]}+\frac{\frac{1}{\omega^{[i]}\theta^{+[i]}}}{\sum\limits_{r=I+1}^{n}\frac{1}{\omega^{[r]}\theta^{+[r]}}}\left(K-\sum\limits_{r=I+1}^nX^{[r]}\right)\\
    &=X^{[i]}+\frac{\frac{1}{\omega^{[i]}\theta^{+[i]}}}{\sum\limits_{r=I+1}^{n}\frac{1}{\omega^{[r]}\theta^{+[r]}}}\left(K-\sum\limits_{r=I+1}^nX^{[r]}\right), \quad \forall i=I+1,\cdots,n.\\
    \mu&=\frac{1}{\sum\limits_{r=I+1}^{n}\frac{1}{\omega^{[r]}\theta^{+[r]}}}\left(\sum\limits_{r=I+1}^nX^{[r]}-K\right)\\
    \lambda_1^{[i]}
    &=\frac{1}{\sum\limits_{r=I+1}^{n}\frac{1}{\omega^{[r]}\theta^{+[r]}}}\left(\sum\limits_{r=I+1}^nX^{[r]}-K\right)-\omega^{[i]}\theta^{+[i]}X^{[i]}, \quad \forall i=1,\cdots, I.
\end{align*}
Since $\lambda_1^{[i]}>0$ for $i = 1, \dots, I$ and $0\le K^{[i]}\le X^{[i]}$ for $i = I+1, \cdots, n$, the following conditions hold. $K\leq\sum\limits_{r=I+1}^nX^{[r]}$,
\begin{align*}
    X^{[i]}&<\frac{\frac{1}{\omega^{[i]}\theta^{+[i]}}}{\sum\limits_{r=I+1}^{n}\frac{1}{\omega^{[r]}\theta^{+[r]}}}\left(\sum\limits_{r=I+1}^nX^{[r]}-K\right), \forall i=1,2,\cdots,I,\\
   X^{[i]}&\ge \frac{\frac{1}{\omega^{[i]}\theta^{+[i]}}}{\sum\limits_{r=I+1}^{n}\frac{1}{\omega^{[r]}\theta^{+[r]}}}\left(\sum\limits_{r=I+1}^nX^{[r]}-K\right),\forall i=I+1,\cdots,n.
\end{align*}

\item $\lambda_1^{(i)}=0$ for all $i = 1, \cdots, n$, and $\lambda_2^{(i)}=0$ for all $i = 1, \cdots, n$.

Now it only remains to consider the case in which $\lambda_1^{(i)} = 0$ and $\lambda_2^{(i)} = 0$ for all $i = 1, \dots, n$. Since there is no divergence in the values of $\lambda_1^{(i)}$ and $\lambda_2^{(i)}$, rearrangement does not make a difference. Nevertheless, we will use the ordered indices $[i]$ here for consistency. 

In this case, conditions now become
\begin{align*}
    \omega^{[i]}\theta^{+[i]}(K^{[i]}-X^{[i]})+\mu=0&, \quad \forall i =1, \dots, n\\
    0\le K^{[i]}\le X^{[i]}&, \quad \forall i=1,\dots,n \\
    \sum_{i=1}^{n}K^{[i]}=K&.
\end{align*}
Solving this system gives us 
\begin{align*}
    K^{[i]}&=\left(1-\frac{\frac{1}{\omega^{[i]}\theta^{+[i]}}}{\sum\limits_{r=1}^{n}\frac{1}{\omega^{[r]}\theta^{+[r]}}}\right)X^{[i]}+\frac{\frac{1}{\omega^{[i]}\theta^{+[i]}}}{\sum\limits_{r=1}^{n}\frac{1}{\omega^{[r]}\theta^{+[r]}}}\left(K-\sum\limits_{r\neq i}X^{[r]}\right) \\
    &=X^{[i]}+\frac{\frac{1}{\omega^{[i]}\theta^{+[i]}}}{\sum\limits_{r=1}^{n}\frac{1}{\omega^{[r]}\theta^{+[r]}}}\left(K-\sum_{r=1}^nX^{[r]}\right), \quad \forall i=1,\cdots,n.
\end{align*}
Since we assume there is a system wide shortage, $K-\sum_{r=1}^nX^{[r]}\leq 0$, and thus we get $K^{[i]}\le X^{[i]}$ for all $i = 1,\dots, n$. Then, the only required condition for this result to hold is $0\le K^{[i]}$ for all $i=1,\cdots,n$, which gives us
\[X^{[i]}\ge \frac{\frac{1}{\omega^{[i]}\theta^{+[i]}}}{\sum\limits_{r=1}^{n}\frac{1}{\omega^{[r]}\theta^{+[r]}}}\left(\sum_{r=1}^{n}X^{[r]}-K\right), \quad \forall i=1,2,\cdots,n.\]

This result can actually be seen as a special case of \textbf{Case} \ref{case:3.3}, where $I=0$, and therefore, we can combine the results from these two cases.

\end{casesp}

\end{casesp}

In summary, the result depends on whether there is a system wide surplus or shortage. In the event of a system wide surplus, \textit{i.e.}, $K^{(i)} > X^{(i)}$ for $i = 1, \dots, n$, the optimal allocation in each region is given by
\begin{align*}
    K^{(i)}=\left(1-\frac{\frac{1}{\omega^{(i)}\theta^{-(i)}}}{\sum\limits_{r=1}^{n}\frac{1}{\omega^{(r)}\theta^{-(r)}}}\right)X^{(i)}+\frac{\frac{1}{\omega^{(i)}\theta^{-(i)}}}{\sum\limits_{r=1}^{n}\frac{1}{\omega^{(r)}\theta^{-(r)}}}\left(K-\sum\limits_{r\neq i}X^{(r)}\right).
\end{align*}

In the event of a system wide shortage, \textit{i.e.}, $0 \leq K^{(i)} \le X^{(i)}$ for $i = 1, \dots, n$, it has been demonstrated that the solution can be found by sorting $\lambda_1^{(i)}$ in such a way that the first $I=0,1,\dots,n-1$ of them are greater than $0$. If $I$ is $0$, then all $\lambda_1^{(i)} = 0$. The $I$ here should satisfy the following conditions, $K\leq\sum\limits_{r=I+1}^nX^{[r]}$,
{\allowdisplaybreaks
\begin{align*}
    X^{[i]}&<\frac{\frac{1}{\omega^{[i]}\theta^{+[i]}}}{\sum\limits_{r=I+1}^{n}\frac{1}{\omega^{[r]}\theta^{+[r]}}}\left(\sum\limits_{r=I+1}^nX^{[r]}-K\right), \forall i=1,2,\cdots,I,\\
   X^{[i]}&\ge \frac{\frac{1}{\omega^{[i]}\theta^{+[i]}}}{\sum\limits_{r=I+1}^{n}\frac{1}{\omega^{[r]}\theta^{+[r]}}}\left(\sum\limits_{r=I+1}^nX^{[r]}-K\right),\forall i=I+1,\cdots,n,
\end{align*}
where if $I=0$, then the first inequality can be discarded. Once $I$ is identified, the optimal allocation in each region is given by $K^{[1]}=\dots=K^{[I]}=0$,}
\begin{align*}
K^{[i]}&=\left(1-\frac{\frac{1}{\omega^{[i]}\theta^{+[i]}}}{\sum\limits_{r=I+1}^{n}\frac{1}{\omega^{[r]}\theta^{+[r]}}}\right)X^{[i]}+\frac{\frac{1}{\omega^{[i]}\theta^{+[i]}}}{\sum\limits_{r=I+1}^{n}\frac{1}{\omega^{[r]}\theta^{+[r]}}}\left(K-\sum\limits_{r=I+1}^nX^{[r]}\right), \quad \forall i=I+1,\cdots,n.
\end{align*}
Since $\lambda_1^{(i)}$ are sorted, the final result should have been re-sorted accordingly. However, the conditions are the same regardless of the ordering. Finally, we relabel $I$ and $n-I$ for the sake of simplifying the notations in the main text, and hence $X^{(i)}$ should be sorted in descending order instead.

\end{proof}

\section{Parameter Values in Numerical Examples} \label{sec:para}

This section offers an inventory of all model parameters used in earlier sections. The same sets of parameters are used for all calculations in the three-pillar framework.


\begin{table}[h!]
\centering
\begin{tabular}{lc}
\hline
\textbf{Parameter}                                           & \textbf{Value} \\ \hline
Percentage of intensive care patients requiring ventilators (class $I_3$) & 0.9   \\
Units of PPE required per exposed patient (class $E$)             & 5     \\
Units of PPE required per hospitalized patient (class $I_2$)            & 15    \\
Units of PPE required per intensive care patient (class $I_3$)            & 20    \\ \hline
\end{tabular}
\caption{Demand assessment parameters; values are chosen in the ranges provided in Tables \ref{tbl:vent} and \ref{ECDCTable2}.}
\end{table}

\begin{table}[h!]
\centering
\begin{tabular}{ll}
\hline
\textbf{Parameter}                                                                  & \textbf{Value}                     \\ \hline
Participating states                                                                & New York, Florida, California      \\
Cost of possession per unit per day ($c_j$)                                               & 1                                  \\
Initial stockpile cost per unit ($c_0$)                                             & 25120                              \\
Daily production rate ($a$)                                                         & 10 Units                           \\
Shortage/surplus cost ($\theta_j^+$/$\theta_j^-$) in Figure \ref{fig:aggre_vent_supply_optim}  & 1000/1000                             \\
Shortage/surplus cost ($\theta_j^+$/$\theta_j^-$) in Figure \ref{fig:aggre_vent_supply_optim_asym}  & 1000/20                            \\
Time varying weight ($\omega_j$)                                                    & Proportional to daily demand $X_j$ \\ \hline
\end{tabular}
\caption{Ventilator planning parameters \citep{Porpora2020,Rowland2020, Patel2020}}
\end{table}



\begin{table}[H]
\centering
\begin{tabular}{ll}
\hline
\textbf{Parameter}                                                                  & \textbf{Value}                     \\ \hline
Participating states                                                                & New York, Florida, California      \\
Cost of possession per 1000 units per day ($c_j$)                                         & 0.01                                  \\
Initial stockpile cost per 1000 units ($c_0$)                                       & 0.5                                  \\
Daily production rate ($a$)                                                         & 50000 units                        \\
Shortage cost ($\theta_j^+$) & 1                                  \\
Time varying weight ($\omega_j$)                                                    & Proportional to daily demand $X_j$ \\ \hline
\end{tabular}
\caption{Personal protective equipment planning parameters}
\end{table}

\begin{table}[H]
\centering
\begin{tabular}{ll}
\hline
\textbf{Parameter}                                                                            & \textbf{Value}                                                  \\ \hline
Participating states                                                                          & New York, Florida, California                                   \\
Shortage/surplus cost ($\theta_j^{(i)+}$/$\theta_j^{(i)-}$) & 1                                                               \\
Weight for resources allocation in region $i$ at time $j$ ($\omega_j^{(i)}$)                  & Proportional to $\sum_{t = j}^m X_t^{(i)}$ \\ \hline
\end{tabular}
\caption{Ventilator/Personal protective equipment allocation parameters}
\end{table}


}

\end{document}